\documentclass[review,11pt,authoryear,3p]{elsarticle}

\usepackage{amssymb}
\usepackage{amsthm}
\usepackage{amsmath}
\usepackage{algorithm}
\usepackage{algpseudocode}
\usepackage{tikz}
\usetikzlibrary{positioning, shapes, decorations.pathreplacing, decorations.markings, arrows, tikzmark, calc}
\usepackage[font=small]{subcaption}

\newtheorem{theorem}{Theorem}
\newtheorem{corollary}{Corollary}
\newtheorem{lemma}{Lemma}

\newtheorem{claim}{Claim}

\newtheorem{definition}{Definition}
\newtheorem{remark}{Remark}

\newtheorem{assumption}{Assumption}

\newcommand{\isit}[1]{{#1}}
\newcommand{\revision}[1]{{#1}}
\newcommand{\secondrevision}[1]{{#1}}

\newcommand{\thirdrevision}[1]{{#1}}

\journal{European Journal of Operational Research}

\begin{document}

\begin{frontmatter}


\author[a]{Akhil Bhimaraju\corref{cor1}}
\ead{akhilb3@illinois.edu}
\author[a]{S. Rasoul Etesami}
\ead{etesami1@illinois.edu}
\author[a]{Lav R. Varshney}
\ead{varshney@illinois.edu}

\tnotetext[ab]{This work was supported in part by NSF grants ECCS-2033900 and EPCN-1944403, the Center for Pathogen Diagnostics through the ZJU-UIUC DESIRE enterprise, and the Air Force Office of Scientific Research under grant FA9550-23-1-0107.}

\cortext[cor1]{Corresponding author}
\affiliation[a]{organization={University of Illinois Urbana-Champaign},
              city={Urbana},
              state={IL},
              postcode={61801},
            country={USA}
            }

\title{Dynamic Batching of Online Arrivals to Leverage Economies of Scale}




\begin{abstract}
Many settings, such as matching riders to drivers in ride-hailing platforms or in-stream video advertising,  require handling arrivals over time. In such applications, it is often beneficial to group the arriving orders or requests into batches and process the larger batches rather than individual arrivals. However, waiting too long to create larger batches incurs a
waiting cost for past arrivals. On the other hand, processing the arrivals too soon leads to higher processing
costs by missing the economies of scale of grouping larger numbers of arrivals
into larger batches. 
\secondrevision{Moreover, the timing of the next arrival is often unknown, meaning fixed-size batches or fixed \thirdrevision{waiting} 
times tend to be poor choices.} 
\revision{In this work, we consider the problem of finding the optimal batching schedule to minimize the 
\thirdrevision{sum of waiting time and processing cost} 
under both offline and online settings.}
In the offline problem in which all arrival times are known a priori,  
we show that the optimal batching schedule can be found in polynomial time by reducing it to a shortest path problem on a weighted acyclic graph. 
For the online problem with unknown arrival times, we develop algorithms that are provably competitive for a broad range of processing-cost functions. 
We also provide a 
lower bound on the competitive ratio that no online algorithm can beat.
Finally, we run numerical experiments on simulated and real data
to demonstrate the effectiveness of our algorithms against the offline benchmark.
\end{abstract}




\end{frontmatter}


\section{Introduction}
\label{sec:introduction}

There are many applications where a service provider processes requests that arrive over time. The service provider can often benefit from economies of scale \citep{Smith1776} if they wait to accumulate a large number of requests. However, waiting too long incurs a higher \emph{waiting cost} on the requests that arrived earlier and are waiting in the queue to be processed. This leads to a trade-off: either process arrivals soon after they arrive and incur a low waiting cost or wait longer to accumulate larger numbers of arrivals, incurring a higher waiting cost but lowering the \emph{processing cost} via economies of scale.
Further, how long the service provider needs to wait until the next arrival is not always known. For instance, the arrival process is often stochastic, with potentially difficult-to-characterize statistics.
This necessitates a principled approach to determining when the service provider should process batches of arrivals and whether waiting longer for more arrivals would be beneficial to serve the requests at a lower cost.
There are many motivating applications that one can consider, and the following are only a few examples.

\subsection*{Matching riders to drivers in ride-hailing platforms}

While early ride-hailing platforms matched riders to the closest driver as soon as they requested a ride, this led to inefficiency in the overall matching. The following is a quote from \citet{Uber2023}:
\begin{quote}
    ``But if we wait just a few seconds after a request, it can make a big difference. It's enough time for a batch of potential rider-driver matches to accumulate. The result is better matches, and everyone's collective wait time is shorter.''
\end{quote}
A similar batched matching framework has also been used for other ride-hailing platforms, such as Didi Chuxing \citep{ZhangHMWZFGY2017}.
This idea also applies to ride-\emph{sharing}, where a larger batch of people needing shared rides can be more efficiently matched with each other \citep{Lyft2016}. More recently, \citet{XieMX2022} have studied how a fixed delay in making matching decisions can improve the overall matching efficiency. 

\secondrevision{While the aforementioned works study the problem of how to match users optimally,
either after a fixed delay or with a fixed batch size, this can be suboptimal in general.}
For fixed-size batches, the final few arrivals might take a very long time to arrive,
incurring a very high waiting cost for the arrivals of the past.
This means we need to process the batch without waiting too long for a fixed number of arrivals.
Similarly for fixed-delay scheduling, if we already have a very large number of arrivals, 
making them wait for a fixed amount of time is suboptimal, and it is instead better to process
them after a relatively shorter wait.
Depending on the number of arrivals and the processing-cost function, fixed-size batches and
fixed-delay scheduling can even incur an unbounded gap with the optimal schedule that processes batches
of any size at any time.%
\footnote{\secondrevision{For example, consider fixed-size batches of size $k$ and a processing-cost function that is square root in the number of arrivals.
If $nk$ samples arrive in quick succession, the optimal cost would be $\sqrt{nk}$, while the fixed-size batch algorithm would incur a cost of $n\sqrt{k}$.
This leads to a ratio of $\sqrt{n}$, which can be arbitrarily large as $n$ increases.}}

\subsection*{In-stream video advertising}

At any moment, online video platforms that display ads need to match the set of online users with the set of ads they wish to serve. 
Since businesses that advertise on the platform have finite budgets, the number of times a particular ad can be displayed is also limited. 
Batching a larger number of users (i.e., consumers of video) and matching the set of available ads to them is more efficient compared to 
deciding which ad to serve to a user as soon as they arrive. 
However, using larger batches requires the platform to wait longer for more arrivals, which incurs a latency cost. 
\revision{While \citet{FengN2022} study the batching vs. inefficiency trade-off in this problem and theoretically characterize the efficiency gains of using larger batches, it is not clear why constant-sized batches must be used or how long one should wait before deciding to perform the matching.}
In this work, we aim to address when it is better to perform the matching
and when it is better to wait longer for more arrivals. 
By carefully managing the waiting time and processing cost, we achieve a constant competitive ratio with respect to the 
optimal offline batching schedule computed in hindsight with full information about future arrivals (see Sec.~\ref{sec:4-comp}).
\revision{As we have seen in the preceding example, this is not possible for constant-sized batches,
which can incur an unbounded gap with respect to the optimal offline cost.}

\subsection*{Real-time processing of IoT data streams}

IoT (Internet of Things) devices generate vast amounts of data, often in multiple streams from different sensors and devices \citep{AhmedYHKAIV2017}.
The data can be processed either using edge computing in the sensing device, in a collection of closely related devices using fog computing,
or in a remotely-located powerful compute cluster using cloud computing.
Processing each data point as it arrives can be inefficient due to the overhead of resource initialization and data transfer.
Further, frameworks based on MapReduce \citep{DeanG2008} like Apache Spark \citep{KarauKWZ2015} tend to be more efficient when processing larger batches of data.
However, waiting too long for sample arrivals introduces latency into the system, which can be undesirable, but not necessarily fatal \citep{FerrariSBDR2017}.
In these applications, there is a need to \thirdrevision{trade off} 
the benefits of batch processing with the downsides of increased latency in a principled manner.
We believe our framework is a step in this direction.

\subsection{Other related work}

Our model is a generalization of the one developed by \citet{BhimarajuV2022}, who study this problem specific to a group-testing application.
Group testing was first developed to efficiently test for syphilis among soldiers by pooling samples from multiple individuals \citep{Dorfman1943}.
If a pooled sample tests negative, every individual in the pool can be declared to be free of the infection.
By pooling intelligently, the number of tests needed to test a given group of individuals can be much smaller than naively testing
every one of them separately \citep{ChanJSA2014, AldridgeBJ2014, ScarlettJ2020, AldridgeJS2019}.
\revision{While our mathematical model has some similarity to that of \citet{BhimarajuV2022}, our results are significantly better:
our model is much more general, our algorithm's performance guarantees are better, and our lower bound that no online algorithm can beat is higher.}
Our analysis also improves the guarantee on the algorithm of \citet{BhimarajuV2022} over what they show in their work.
Moreover, we give a polynomial-time offline optimal algorithm for computing the best batching schedule, which is absent
from \citet{BhimarajuV2022}.

A related line of work is batch-service queuing 
\citep{DebS1973, NeutsC1987, BarlevPPSV2007, ClaeysWLB2010, ClaeysSWLB2013, ChakravarthySR2021},
where large enough batches of Poisson arrivals are accumulated before processing.
While prior work has largely focused on finding the optimal batch size for a given arrival process,
we make no assumptions on the arrival statistics, and our algorithms can adapt to any static or time-varying arrival process.

Our mathematical model is also a generalization of the TCP acknowledgement problem \citep{DoolyGS2001, KarlinKR2001}.
TCP is a widely-used transport-layer protocol on the internet and requires an acknowledgement to be sent for each
received TCP packet.
However, multiple packets can be acknowledged together, which makes it possible to reduce the bandwidth consumed by waiting
for future packets and acknowledging a larger batch of packets together.
However, TCP employs a congestion-control algorithm \citep{Stevens1997}, and waiting too long to acknowledge can make the algorithm perform suboptimally.
The TCP acknowledgement problem is a trade-off between minimizing the bandwidth consumed, by sending fewer acknowledgements, and
not waiting too long to send an acknowledgement for each packet.
Since any number of acknowledgements can be grouped together, the ``processing-cost function'' in this case is a constant.
Our model, on the other hand, generalizes this to any arbitrary submodular function (please see Sec.~\ref{sec:model} for the details),
and our performance guarantees and lower bound are expressed in terms of a parameter that captures this function's curvature.

\revision{Relatedly, burn-in during the testing of semiconductor chips has been modeled as a batching problem \citep{LeeUM1992}.
The approach we use for computing the optimal offline schedule (Sec.~\ref{sec:offline}) bears similarity to the ones used in this context
\citep{LiL1997}. 
\secondrevision{While the offline problems considered here minimize a function of ``tardiness,'' online algorithms have been developed for minimizing the makespan \citep{ZhangCW2001, DengPZ2003, YuanLTF2009best}.} 
However, this does not generalize to our setting, where we consider the objective as the sum of the \thirdrevision{waiting} 
times and the processing cost. Further, our analysis works for arbitrary submodular processing costs, while tardiness and makespan are functions of a scalar processing time associated with each task.%
\footnote{\revision{This problem turns out to be NP-hard under relatively mild assumptions such as unsorted processing times \citep{LiL1997}.}}
Please see \citet{PottsK2000} for a review of variants of the offline batching problem.}

The techniques we use in our analysis bear similarity to those used to analyse job scheduling \citep{Graham1966, HallSSW1997},
where the objective is to minimize a function of the completion times of the jobs on a given number of processors.
``Speed scaling'' is an extension of the job scheduling problem which aims to minimize a combination of the \thirdrevision{waiting} 
times plus
the total energy consumed \citep{AlbersF2007, BansalPS2010, DevanurH2018}, which is
\isit{
similar to our formulation of minimizing the \thirdrevision{waiting} 
times plus the processing cost.}
\revision{However, practical considerations suggest a submodular processing cost (see Sec. \ref{sec:model}), which is incurred at discrete time instants when batching occurs, making the convex programs and primal-dual analysis, such as those used by \citet{HallSSW1997} and \citet{DevanurH2018}, not immediately applicable to our setting.}

\subsection{Contributions}
Motivated by the above and many other applications, here we formulate a general problem in the form of minimization 
of the average \thirdrevision{waiting} 
time of the arrivals plus the average processing cost. 
\revision{We aim to develop online algorithms that assume no information about future arrivals and make irrevocable batching decisions
while updating the known information as new requests arrive.} %
\revision{To evaluate our online algorithm's performance, we use the competitive ratio, 
which is the worst-case cost ratio between the online algorithm and the optimal offline schedule
computed in hindsight with full information about all the arrival times.}
\isit{Thus, if an online algorithm has a certain competitive ratio, then the cost of running the algorithm is never worse than this ratio times the optimal offline cost. 
This also means that our performance guarantee holds in an adversarial setting regardless of potentially time-varying statistics of the arrival process.} 
Moreover, our competitive ratio does not depend on the total number of arrivals and holds for every possible instance of the problem. 
The main contributions of our work are summarized below.
\begin{itemize}
\item \revision{We provide a general formulation to characterize the trade-off between the waiting time and the processing cost in batching problems with online arrivals.} 
\item \revision{We show that the offline problem, in which all the arrival times are known a priori, can be solved in polynomial time by interpreting it as a shortest path problem on a graph (Theorem~\ref{thm:offline}).}
\item We develop a constant competitive algorithm for the online problem without the knowledge of future arrivals (Theorem~\ref{thm:wte} and Corollary~\ref{cor:wte}).
\revision{As far as we know, this is the first work to consider a general submodular function for the processing cost in scheduling problems of the type considered here.}
\item We provide a lower bound for the competitive ratio of any online algorithm (Theorem \ref{thm:lower-bound}). 
\item We evaluate the performance of our online algorithm using extensive numerical experiments.   
\end{itemize}


\subsection{Organization}
The remainder of this paper is organized as follows. Sec.~\ref{sec:model} describes the system model, objective, and assumptions
on the processing-cost function. Sec.~\ref{sec:offline} gives a polynomial-time offline algorithm for computing
the best batching solution in hindsight, using full information about all the 
arrival times. This serves as a benchmark for online algorithms and also helps to evaluate the performance of online algorithms numerically in a reasonable computation time. Sec.~\ref{sec:4-comp} presents our online algorithm, which balances the
waiting and processing costs to achieve a small competitive ratio.
Sec.~\ref{sec:lower-bound} gives a lower bound on the competitive ratio that any online algorithm might achieve. Sec.~\ref{sec:simulations} shows the empirical performance of our algorithms using numerical experiments,
and Sec.~\ref{sec:conclusion} concludes the paper by identifying some future research directions.

\section{Model}
\label{sec:model}

\revision{We consider a system that has to process a sequence of $n$ \emph{samples} that arrive at times \(a_1 \le a_2 \le \cdots \le a_n\),
where $a_i\ge0$ for all $1\le i\le n$.}
\revision{Sample $i$, which arrives at time $a_i$, is associated with a feature vector $v_i\in\mathcal{F}$ 
drawn from some finite set $\mathcal{F}$, which determines the eventual cost of processing that sample.}
For example, in the ride-hailing application, the feature vector could be the current location, rating, and destination of the rider requesting a ride.
\isit{
Using an algorithm \(\textsc{alg}\), the samples are processed in \(m^\textsc{alg}\) batches 
\(\mathcal{B}^\textsc{alg}_1,
\mathcal{B}^\textsc{alg}_2, \ldots, \mathcal{B}^\textsc{alg}_{m^\textsc{alg}}\) at times 
\(s^\textsc{alg}_1< s^\textsc{alg}_2< \cdots< s^\textsc{alg}_{m^\textsc{alg}}\), where
\(\mathcal{B}^\textsc{alg}_j\) is the set of samples processed together at time \(s^\textsc{alg}_j\),
with \(\mathcal{B}^\textsc{alg}_i\cap\mathcal{B}^\textsc{alg}_j=\emptyset\) for
\(i \neq j\).
\revision{The non-intersection condition ensures that no sample is processed more than once.}
For each sample \(i\), define \(d^\textsc{alg}_i\) as the time when it gets processed by \(\textsc{alg}\), i.e., 
\(d^\textsc{alg}_i = s^\textsc{alg}_j\) for \(i\in\mathcal{B}^\textsc{alg}_j\).
Note that a sample can only be processed after it has arrived. Therefore, for \(\textsc{alg}\)
to be valid, we must have \(d^\textsc{alg}_i \ge a_i\) for all \(1\le i \le n\).}
Let \(f:2^\mathcal{F}\mapsto\mathbb{R}^+\cup\{0\}\) be a set function such that $f(\mathcal{X})$ denotes the cost
incurred when we process samples with the set of feature vectors \(\mathcal{X}\subseteq\mathcal{F}\) together in a single batch.%
\footnote{\revision{Note that if two samples have identical feature vectors, $\mathcal{X}$ cannot be a simple set union of the feature vectors, since both samples need to be processed. Therefore, $\mathcal{X}$ must denote a multiset rather than a set. For ease of exposition, we still use the standard set notation, with the understanding that multisets do not necessarily share all the properties of sets. For example, $\mathcal{X} \cup \mathcal{X} \neq \mathcal{X}$.
}}
We expect the processing of a larger batch to have a lower per-sample processing cost
than multiple smaller batches due to economies of scale.
\revision{Further, we expect the processing of a larger batch to incur a higher total cost than a smaller subset of that batch.}
These properties lead us to the following assumption.

\begin{assumption}
We assume that the cost function \(f:2^\mathcal{F}\mapsto\mathbb{R}^+\cup\{0\}\) satisfies the following conditions:
\begin{itemize}
\item[(i)] \(f(\emptyset)=0\);
\item[(ii)] \(f(\mathcal{Y})\le f(\mathcal{X})\) if \(\mathcal{Y}\subseteq\mathcal{X}\) for $\mathcal{X},\mathcal{Y}\subseteq\mathcal{F}$;
\item[(iii)] $f(\mathcal{X}\cup\mathcal{Y})\le f(\mathcal{X})+f(\mathcal{Y})$ for $\mathcal{X},\mathcal{Y}\subseteq\mathcal{F}$.
\end{itemize}
\label{assm:f-assumptions}
\end{assumption}

For a given problem instance, the cost of processing samples $\{k,k+1,\ldots,k+r-1\}$ together
is equal to $f(\{v_k,v_{k+1},\ldots,v_{k+r-1}\})$.
To simplify notation, we also use $f(\{k,k+1,\ldots,k+r-1\})$ to denote this quantity when the meaning is clear.
\isit{We define the cost of algorithm \(\textsc{alg}\) on 
an input instance, denoted by \(J^\textsc{alg}\), as the per-sample average%
\footnote{\revision{We could alternatively use the total cost of the algorithm rather than the average, but the average allows us to directly compare the objective cost on problem instances with different numbers of samples.}} 
of the total waiting time 
of all the samples plus the processing cost, i.e.,\footnote{Note that we only consider the time the sample is kept waiting
until the processing starts.
Any ``waiting time'' during the processing can be included in the
cost function \(f\).
}}
\begin{align}
J^\textsc{alg} = 
\frac{1}{n}\left(\sum_{i=1}^n(d_i-a_i)+\sum_{j=1}^{m^\textsc{alg}}f\left(
\mathcal{B}^\textsc{alg}_j\right)\right).
\label{eq:jalg-definition}
\end{align}

To ensure that the processing schedule is equitable to the samples (and for mathematical tractability),
we assume that the arrivals are processed in order, i.e., a sample which arrived before another sample must be
processed no later than the latter sample.
We state this formally as Assumption~\ref{assm:consecutive-sets}.
\begin{assumption}
    A valid algorithm \textsc{alg} processes samples $i_1$ and $i_2$ such that if $a_{i_1}\le a_{i_2}$,
    then $d_{i_1}^\textsc{alg}\le d_{i_2}^\textsc{alg}$, i.e., for $i_1\in\mathcal{B}^\textsc{alg}_{j_1}$ and $i_2\in\mathcal{B}^\textsc{alg}_{j_2}$,
    we have $s_{j_1}^\textsc{alg}\le s_{j_2}^\textsc{alg}$.
    \label{assm:consecutive-sets}
\end{assumption}

\isit{An algorithm \textsc{alg} is \emph{online} if in determining to 
process the batch \(\mathcal{B}^\textsc{alg}_j\)
at time \(s^\textsc{alg}_j\), \textsc{alg} makes no use of information about any 
future samples in \(\{i:a_i>s^\textsc{alg}_j\}\), nor can it change its past decisions.}
In other words, an online algorithm \textsc{alg} makes its decisions irrevocably and purely based on the past observed samples. Moreover, we let \(\textsc{opt}\) denote the optimal \emph{offline} algorithm that selects the batches \(\mathcal{B}\) to minimize the objective cost \eqref{eq:jalg-definition} assuming full knowledge of the arrival times. We denote the cost of the optimal offline algorithm by \(J^\textsc{opt}=\min_\textsc{alg}J^\textsc{alg}\). 

To analyze the performance of our devised online algorithms, as is conventional in the literature of online optimization, we adopt the notion of competitive ratio, as defined next.

\begin{definition}
An online algorithm \(\textsc{alg}\) is said to be \(\rho\)\emph{-competitive} for \(\rho\ge1\) if the supremum of \(\frac{J^\textsc{alg}}{J^\textsc{opt}}\) over all problem instances \(\{(a_1,v_1), \ldots, (a_n,v_n)\}\) for all
\(n\) is less than or equal to \(\rho\), in which case we refer to \(\rho\) as a  \emph{competitive ratio} of \textsc{alg}.
\end{definition}

Our main objective in this paper is to design online algorithms that provably admit a small (and constant) competitive ratio and
to establish nearly matching lower bounds on the competitive ratio that no online algorithm can 
\secondrevision{beat.}
To that end, in the next section, we first consider the offline problem with known arrival times, which serves as the benchmark for our subsequent competitive ratio analysis.

\section{The Offline Problem}
\label{sec:offline}

In this section, we consider the offline version of the problem, in which all the samples' arrival times $a_1,a_2,\ldots,a_n$
and their corresponding feature vectors $v_1,v_2,\ldots, v_n$ are known, and we wish to compute the optimal schedule that minimizes the cost given by \eqref{eq:jalg-definition}.
\revision{Developing an effective algorithm for the offline problem helps us empirically evaluate the performance of our online algorithms, as we can then compute the best (offline) schedule in reasonable time and estimate the competitive ratio through simulation.}
In fact, for some practical applications such as group testing 
\revision{\citep{BhimarajuV2022},} 
it is not implausible that the testing center might know when the samples might arrive if people have booked testing slots in advance. In this case, one can use the offline algorithm to compute the best batching schedule and run it as the samples arrive at their given times.%
\footnote{A related application in packing shipping containers has been considered by \citet{WeiKJ2021}.}
While we can find the optimal offline schedule using brute force by computing every possible schedule and its corresponding cost,
this is an extremely inefficient solution that requires exponential computational time. While it is tempting to infer that solving the optimal offline problem is NP-hard, as for similar problems \citep{Ullman1975,GareyJ1975}, we show this is not the case here.%
\footnote{Assuming the function $f(\mathcal{X})$ can be computed in time polynomial in $|\mathcal{X}|$.}

We propose the \textsc{OffShortPath} (\textsc{OSP}) algorithm 
based on \citet{ChakravartyOR1982} for the offline problem.
While the \emph{concave in subset-sum for fixed cardinality} property of \citet{ChakravartyOR1982} does not appear to be applicable here, we do not require this property as we see in the proof of Theorem~\ref{thm:offline}.
The offline algorithm starts by constructing a weighted directed graph $G=(Q,E)$ with $|Q|=n+1$ nodes, where
node $q_i\in Q$ corresponds to sample $i$ (the sample which arrived at time $t=a_i$) for $i\in\{1,2,\ldots,n\}$.
We also include an additional node $q_{n+1}$, which denotes the end state after all the samples
have arrived (its purpose will become clear in Alg.~\ref{alg:osp}).
A weight of $e_{ij}$ is assigned to edge $(q_i,q_j,e_{ij})\in E$ for $i<j$ as follows:
\begin{align}
    e_{ij} := f(\{v_i,v_{i+1},\ldots,v_{j-1}\}) + \sum_{k=i}^{j-1} (a_{j-1}-a_k)
        \quad \text{for}\ i < j.
\label{eq:edge-weight}
\end{align}
Effectively, $e_{ij}$ can be thought as the unnormalized cost of processing the samples
$\{i,{i+1},\ldots,{j-1}\}$ in a single batch at time $t=a_{j-1}$.
The solution to the optimal schedule is given by computing the
minimum weight path (a.k.a. shortest path) from $q_1$ to $q_{n+1}$ in this weighted acyclic graph.
If the minimum weight path is given by $(q_1=:q_{r_1},q_{r_2},\ldots,q_{r_p}:=q_{n+1})$,
the optimum solution is given by processing in $p-1$ batches, where
the $j$th batch $\mathcal{B}_j:=\{{r_j},{r_j+1},\ldots,{r_{j+1}-1}\}$
is processed at $d_j=a_{r_{j+1}-1}$.
\secondrevision{The minimum weight path can be computed using Dijkstra's algorithm or a simpler topological sort
followed by a dynamic program \citep{CormenLRS2022}.}
We state this algorithm formally as Alg.~\ref{alg:osp}, whose computational complexity is analyzed in Theorem~\ref{thm:offline}.

\begin{algorithm}[h!]
\caption{\textsc{OffShortPath} (\textsc{OSP})}
\label{alg:osp}
\textbf{Input:} \(\{a_1,a_2,\ldots,a_n\}\), \(f:2^\mathcal{F}\mapsto\mathbb{R}^+\cup\{0\}\)\\
\textbf{Output:} $m$, $(\mathcal{B}_1,\mathcal{B}_2,\ldots,\mathcal{B}_m)$, $(s_1,s_2,\ldots,s_m)$\\
\textbf{Initialize:} Nodes $Q=\{q_1,q_2,\ldots,q_{n+1}\}$, edges $E=\emptyset$
\begin{algorithmic}[1]
\State Construct the weighted acyclic graph $G=(Q,E)$:
\For {$i<j$ \textbf{in} $(q_i,q_j)\in Q\times Q$}
    \State $e_{ij}\gets f(\{v_i,v_{i+1},\ldots,v_{j-1}\}) + \sum_{k=i}^{j-1} (a_{j-1}-a_k)$
    \State $E\gets E\cup\{(q_i,q_j,e_{ij})\}$
\EndFor
\State $(q_{r_1},\ldots,q_{r_p})\!\gets\!
            \textsc{MinWeightPath}(G,\ \text{from}\ q_1\ \text{to}\ q_{n+1})$
\State $m\gets p-1$
\For {$j\in\{1,2,\ldots,m\}$}
    \State $\mathcal{B}_j\gets \{r_j,r_j+1,\ldots, r_{j+1}-1\}$
    \State $s_j\gets a_{r_{j+1}-1}$
\EndFor
\end{algorithmic}
\end{algorithm}

\begin{theorem}
The algorithm \textsc{OffShortPath} (Alg.~\ref{alg:osp})
gives the optimum offline schedule that minimizes the cost given by \eqref{eq:jalg-definition}.
Further, it can be implemented with a time complexity of $O(n^2T_n)$,
where $n$ is the number of samples and the time complexity of $f(\mathcal{X})$ is $T_{|\mathcal{X}|}$.
\label{thm:offline}
\end{theorem}
\begin{proof}[Proof.]
Let us first observe that the construction of the graph in Alg.~\ref{alg:osp}
always results in a connected (directed) acyclic graph.
This is because we add edges from $q_i$ to $q_j$ only if $i<j$ (which makes the
graph acyclic), and we add edges for all such $(i,j)$ (which makes the graph connected).
Let an optimal schedule for the problem instance be given by the $m^\textsc{opt}$ batches
$\mathcal{B}_1^\textsc{opt},\mathcal{B}_2^\textsc{opt},\ldots,\mathcal{B}_{m^\textsc{opt}}^\textsc{opt}$
processed at times $s_1,s_2,\ldots,s_{m^\textsc{opt}}$ respectively.
Define $u_j^\textsc{opt}$ for $j\in\{1,2,\ldots,m^\textsc{opt}\}$ as
\begin{align*}
    u_j^\textsc{opt} = q_{\min\{i\mathop{:}i\in\mathcal{B}_j^\textsc{opt}\}}.
\end{align*}
By Assumption~\ref{assm:consecutive-sets}, the optimal schedule is
``consecutive,'' i.e.,
\begin{align*}
    x\in\mathcal{B}_i^\textsc{opt}\ \text{and}\ y\in\mathcal{B}_j^\textsc{opt}\ 
    \text{for}\ i<j \implies x < y\ \forall\ x,y.
\end{align*}
This implies $(u_1^\textsc{opt},u_2^\textsc{opt},\ldots,u_{m^\textsc{opt}}^\textsc{opt},q_{n+1})$
forms a valid (directed) path in the graph $G$.%
\footnote{Note that $u_1^\textsc{opt}=q_1$ by definition.}
Further, we can see that the weight of this path is equal to $n$ times
$J^\textsc{opt}$, the cost of running the optimal schedule.%
\footnote{In an optimal schedule, there is no waiting time after the last sample
in the batch has arrived.}
Using similar arguments, every path in the graph from $q_1$ to $q_{n+1}$
corresponds to a schedule which satisfies Assumption~\ref{assm:consecutive-sets}.
\revision{
    By this equivalence,
    the output of Alg.~\ref{alg:osp} gives us an optimal schedule.
}

\revision{
    In the runtime complexity, the graph construction takes $O(n^2T_n)$ time using a cumulative sum array $S_j=\sum_{k=1}^ja_k$.}
Computing the minimum weight path in a weighted directed acyclic graph can be done in $O(|Q|+|E|)$ using topological sorting \citep{CormenLRS2022},%
\footnote{In fact, we can avoid the sorting step in our case since $(q_1,q_2,\ldots,q_{n+1})$
is a natural topological order.}
which has a complexity of $O(n^2)$ in our problem. This gives us an overall complexity of $O(n^2+n^2T_n)=O(n^2T_n)$.
\end{proof}


\revision{
    Reducing the batch-scheduling problem to a minimum-weight path problem implies that it can also be written as a linear program. Interestingly, solving the dual of this program reduces to solving a dynamic program very similar to the standard one used to find the shortest path in a directed acyclic graph.}
We refer to \ref{sec:linear-program-offline} for more details.
\section{Competitive Online Algorithms}
\label{sec:4-comp}

\secondrevision{We now present our online algorithms for adaptive batching 
and prove their performance guarantees in terms of the competitive ratio.}
\isit{Let us first introduce some useful notation. For a given problem instance and an algorithm \(\textsc{alg}\) (either online or offline), 
let \(W^\textsc{alg}\) and \(F^\textsc{alg}\) denote the average \thirdrevision{waiting} 
time and the average (per-sample) processing cost, respectively:
\begin{align*}
W^\textsc{alg} = \frac{1}{n}\sum_{i=1}^n(d^\textsc{alg}_i-a_i),
\end{align*}
}%
\vspace{-1.5em}
\begin{align*}
F^\textsc{alg} = \frac{1}{n}\sum_{j=1}^{m^\textsc{alg}}f\left(\mathcal{B}^\textsc{alg}_j\right),
\end{align*}
\isit{where we note that $J^\textsc{alg} = W^\textsc{alg} + F^\textsc{alg}$. Moreover, let \(u^\textsc{alg}_t\)
denote the number of samples that have arrived by time \(t\)
but not yet been processed by time \(t\), i.e., 
\begin{align}\nonumber
u^\textsc{alg}_t=\big\lvert\{i:a_i\le t < d^\textsc{alg}_i\}\big\rvert.
\end{align}
Observe that the average \thirdrevision{waiting} 
time can be written as
\begin{align}
W^\textsc{alg} = \frac{1}{n} \int_0^\infty u^\textsc{alg}_\tau d\tau.
\label{eq:wait-time-integral}
\end{align}

\revision{The main trade-off in the batching problem is the following.}
If we wait for a long time before processing, we accumulate many samples, thus decreasing the per-sample processing cost (via batching gains)
but increasing the waiting time. However, suppose we process too aggressively soon after a few samples have arrived. 
In that case, we might lose out on accumulating enough samples and thus incur a high per-sample processing cost even though the waiting
time would be lower.}

We propose the \textsc{WaitTillAlpha} (\textsc{WTA})
algorithm, which balances these two components.
We first compute the cumulative waiting time of arrived samples that are
yet to be processed and compare this with the cost of processing them
all together. Initially, the cumulative \thirdrevision{waiting} 
time would be small, but it would increase as time progresses. Once its value equals $\alpha$ times the cost of processing all these samples together, we process them, where $\alpha$ is a scalar balancing factor that can be optimized to obtain a smaller competitive ratio. We state this formally as Alg.~\ref{alg:wte}, whose performance guarantee is shown in Theorem~\ref{thm:wte}.

{\algrenewcommand\algorithmicthen{{\protect{\textbf{then}\footnotemark}}}
\begin{algorithm}[h!]
\caption{\textsc{WaitTillAlpha} (\textsc{WTA})}
\label{alg:wte}
\textbf{Initialize:} \(b_\text{prev}\leftarrow0\)
\begin{algorithmic}[1]
\State At each time \(t\):
\If {\(\int_{b_\text{prev}}^t \!\!\!u_\tau d\tau = \alpha f\!\left(\{v_i\!:\!a_i\!\in\![0,t]\ \text{and} \ d_i\!\notin\![0,t)\}\right)\)}
    \State Process all the available samples together  at time \(t\)
    \State \(b_\text{prev}\leftarrow t\)
\EndIf
\end{algorithmic}
\end{algorithm}
\footnotetext{For \(b_\text{prev}\neq0\), the right side can be written as \(f\left(\{v_i:a_i\in(b_\text{prev},t]\}\right)\).}
}

\begin{theorem}
The online algorithm \(\textsc{WaitTillAlpha}\) (Alg.~\ref{alg:wte}) admits a competitive
ratio of
\begin{align*}
    \frac{J^\textsc{WTA}}{J^\textsc{opt}}\le \left(1+\frac{1}{\alpha}\right)\max\left\{1,\frac{\alpha}{\Gamma}\right\},
\end{align*}
for any problem instance, where%
\footnote{Recall from Sec.~\ref{sec:model} that $\mathcal{X}$ and $\mathcal{Y}$ are multisets rather
than sets and $\mathcal{X}\cup\mathcal{X}\neq\mathcal{X}$.
We also assume in this definition that $\mathcal{X}$ and $\mathcal{Y}$ are not both simultaneously $0$ to avoid the undefined ratio $0/0$.}
\begin{align*}
\Gamma=\inf_{\mathcal{X},\mathcal{Y}\subseteq\mathcal{F}}\frac{f(\mathcal{X}\cup\mathcal{Y})}{f(\mathcal{X})+f(\mathcal{Y})}.
\end{align*}
\label{thm:wte}
\end{theorem}

The parameter $\Gamma$ captures the curvature or ``submodularity'' of the function $f$.
In fact, Assumption~\ref{assm:f-assumptions} implies that $\Gamma\in[\tfrac{1}{2},1]$.
Depending on what we know about the value of $\Gamma$ for our particular cost function $f$, 
we can choose $\alpha$ in Alg. \ref{alg:wte} appropriately to obtain a smaller competitive ratio (see Corollary~\ref{cor:wte} for more details).
It is straightforward to see that \textsc{WTA} works in an online fashion since computing the set \(u_t\) only requires us to know the
samples that have already arrived but have not yet been processed.

\revision{Next, we state the following corollary that gives us possible values of $\alpha$ that could be used in practice.
\begin{corollary}
The following are true for all functions $f(\cdot)$ that satisfy Assumption \ref{assm:f-assumptions}.
\begin{enumerate}
    \item[(i)] If $\alpha=\frac{1}{2}$, \textsc{WTA} admits a competitive ratio of $3$.
    \item[(ii)] If $\alpha=1$, \textsc{WTA} admits a competitive ratio of $\frac{2}{\Gamma}\leq 4$.
    \item[(iii)] If we know $\Gamma$, we can set $\alpha={\Gamma}$, giving us a competitive ratio of $\left(1+\frac{1}{\Gamma}\right)\in [2,3]$.
\end{enumerate}
\label{cor:wte}
\end{corollary}
\begin{proof}[Proof.]
When $\alpha=\frac{1}{2}$, using $\Gamma\in[\tfrac{1}{2},1]$ gives
\begin{align*}
    \max\left\{1,\frac{\alpha}{\Gamma}\right\}=1.
\end{align*}
Theorem~\ref{thm:wte} then implies that the competitive ratio is $\left(1+\frac{1}{\alpha}\right)=3$.

When $\alpha=1$, we have
\begin{align*}
    \max\left\{1,\frac{\alpha}{\Gamma}\right\}=\frac{1}{\Gamma}.
\end{align*}
So the competitive ratio given in Theorem~\ref{thm:wte} becomes
$\left(1+\frac{1}{\alpha}\right)\frac{1}{\Gamma}=\frac{2}{\Gamma}$.
Since $\Gamma\in[\tfrac{1}{2},1]$, this ratio cannot be more than $4$.

Finally, setting $\alpha={\Gamma}$ implies
\begin{align*}
    \max\left\{1, \frac{\alpha}{\Gamma}\right\}=1,
\end{align*}
which gives us the competitive ratio
$\left(1+\frac{1}{\alpha}\right)=\left(1+\frac{1}{\Gamma}\right)$.
As $\Gamma\in[\tfrac{1}{2},1]$, we have $\left(1+\frac{1}{\Gamma}\right)\in [2,3]$.
\end{proof}}

\revision{In practical scenarios, the exact value of $\Gamma$ might be unknown or difficult to compute, 
especially if $f(\cdot)$ is a complicated function.
Corollary~\ref{cor:wte} implies that the \textsc{WTA} algorithm can achieve a competitive ratio of $3$ by setting $\alpha=\frac{1}{2}$ independent of $\Gamma$.
Thus, we can achieve a competitive ratio of $3$ without having to know the exact value of $\Gamma$.}

Before proving Theorem~\ref{thm:wte}, we first state the following key lemma.

\begin{lemma}
The following is true for the \textsc{WTA} algorithm:%
\footnote{Note that if there are multiple optimal schedules,
all of them should have the same \(J^\textsc{opt}\), but they might have different values for \(W^\textsc{opt}\) and \(F^\textsc{opt}\).
For the purpose of proving
Theorem~\ref{thm:wte}, it is sufficient if Lemma~\ref{claim:testing-cost} holds for any one possible optimal schedule.}
\begin{align*}
\frac{1}{n}\int_0^\infty\left(u^\textsc{WTA}_\tau-u^\textsc{opt}_\tau\right)^+d\tau \le \frac{\alpha}{\Gamma}F^\textsc{opt},
\end{align*}
where \((x)^+=\max\{0,x\}\) and $\Gamma=\inf_{\mathcal{X},\mathcal{Y}\subseteq\mathcal{F}}\frac{f(\mathcal{X}\cup\mathcal{Y})}{f(\mathcal{X})+f(\mathcal{Y})}$.
\label{claim:testing-cost}
\end{lemma}

We defer the proof of Lemma~\ref{claim:testing-cost} to Sec.~\ref{sec:testing-cost}.
First, we show how this lemma can be used to prove Theorem~\ref{thm:wte}.

\begin{proof}[Proof of Theorem~\ref{thm:wte}.]
\isit{Observe that for any algorithm \textsc{alg}, we can write
\begin{align*}
u^\textsc{alg}_\tau \le u^\textsc{opt}_\tau + \left(u^\textsc{alg}_\tau-u^\textsc{opt}_\tau\right)^+
\ \text{for all} \ \tau \ge 0.
\end{align*}
Writing this for \(\textsc{WTA}\) and integrating over time gives us
\begin{align*}
\int_0^\infty u^\textsc{WTA}_\tau d\tau
\le \int_0^\infty u^\textsc{opt}_\tau d\tau
+ \int_0^\infty \left(u^\textsc{WTA}_\tau-u^\textsc{opt}_\tau\right)^+ d\tau.
\end{align*}}
Using \eqref{eq:wait-time-integral} and Lemma~\ref{claim:testing-cost}, we get
\begin{align}
W^\textsc{WTA} \le W^\textsc{opt} + \frac{\alpha}{\Gamma}F^\textsc{opt}
&\le \max\left\{1, \frac{\alpha}{\Gamma}\right\}\left( W^\textsc{opt} + F^\textsc{opt} \right) \nonumber\\
&= \max\left\{1,\frac{\alpha}{\Gamma}\right\}J^\textsc{opt}.
\label{eq:w-less-than-f}
\end{align}

Since \textsc{WTA} processes samples at the time instant when their cumulative \thirdrevision{waiting} 
time
is equal to $\alpha$ times their processing cost, we have
\begin{align*}
W^\textsc{WTA} = \alpha F^\textsc{WTA},
\end{align*}
which gives \(J^\textsc{WTA}=\left(1+\frac{1}{\alpha}\right)W^\textsc{WTA}\).
Using \eqref{eq:w-less-than-f} then gives us
\begin{align*}
J^\textsc{WTA} \le \left(1+\frac{1}{\alpha}\right)\max\left\{1,\frac{\alpha}{\Gamma}\right\}J^\textsc{opt},    
\end{align*}
which concludes the proof.
\end{proof}

We are now ready to prove Lemma \ref{claim:testing-cost}.

\subsection{Proof of Lemma~\ref{claim:testing-cost}}
\label{sec:testing-cost}
Let \textsc{opt} be the optimal schedule.
\isit{For proving Lemma~\ref{claim:testing-cost}, we divide the integral
\(\int_0^\infty (u^\textsc{WTA}_\tau-u^\textsc{opt}_\tau)^+d\tau\)
into segments:
\begin{align}
\int_0^\infty (u^\textsc{WTA}_\tau-u^\textsc{opt}_\tau)^+d\tau
= \sum_{j=1}^{m^\textsc{WTA}} \int_{s^\textsc{WTA}_{j-1}}^{s^\textsc{WTA}_j} (u^\textsc{WTA}_\tau-u^\textsc{opt}_\tau)^+d\tau,
\label{eq:proof-claim-1-eq1}
\end{align}
where we have defined \(s^\textsc{WTA}_0=0\).
Note that for \(\tau>s^\textsc{WTA}_{m^\textsc{WTA}}\),
the integrand is \(0\) since all the samples have been processed by \textsc{WTA},
and so \(u^\textsc{WTA}_\tau=0\).

The set of batches in \textsc{opt} can be partitioned into sets \(\mathcal{S}_j\) based on the
\((s^\textsc{WTA}_{j-1},s^\textsc{WTA}_j]\) segment in which \textsc{opt}
processes them:%
\footnote{
For $j=1$, we have $s_{j-1}^\textsc{WTA}=0$ and the appropriate set here
is $[s_0^\textsc{WTA},s_1^\textsc{WTA}]$.}
\begin{align*}
\mathcal{S}_j = \left\{k:s^\textsc{opt}_k\in(s^\textsc{WTA}_{j-1},s^\textsc{WTA}_j]\right\}\ \text{for} \ j\in\{1,2,\ldots,m^\textsc{WTA}\}.
\end{align*}

Observe that at \(\tau=s^\textsc{WTA}_{j-1}\) for all \(j\), \(u^\textsc{WTA}_\tau=0\) since \textsc{WTA}
processes all the samples available at \(\tau\).
So \((u^\textsc{WTA}_\tau-u^\textsc{opt}_\tau)^+=0\) at \(\tau=s^\textsc{WTA}_{j-1}\).
Further, if \(\mathcal{S}_{j}=\emptyset\), then \textsc{opt} does not process any samples
before \textsc{WTA} processes all the available samples again, and so
\((u^\textsc{WTA}_\tau-u^\textsc{opt}_\tau)^+=0\) for all \(\tau\in(s^\textsc{WTA}_{j-1},s^\textsc{WTA}_{j}]\).
This gives
\begin{align}
\int_{s^\textsc{WTA}_{j-1}}^{s^\textsc{WTA}_j}\left(u^\textsc{WTA}_\tau-u^\textsc{opt}_\tau\right)^+d\tau = 0\ \ \text{if}\ \ \mathcal{S}_j=\emptyset.
\label{eq:partitioned-optimal-cost-empty}
\end{align}

If \(\mathcal{S}_j\neq\emptyset\), we have}%
\begin{align}\label{eq:testing-cost-non-empty-partitions}
\int_{s^\textsc{WTA}_{j-1}}^{s^\textsc{WTA}_j}\left(u^\textsc{WTA}_\tau-u^\textsc{opt}_\tau\right)^+d\tau
\le \int_{s^\textsc{WTA}_{j-1}}^{s^\textsc{WTA}_j} u^\textsc{WTA}_\tau d\tau  = \alpha f\left( \mathcal{B}_j^\textsc{WTA}\right),
\end{align}
where the equality follows directly from the definition of the \textsc{WTA} algorithm.

We now define some useful variables, which are also listed in Table~\ref{tab:proof-variables} for ease of presentation. We refer to Fig.~\ref{fig:proof-variables} for a pictorial illustration of those variables. Let $\tilde{\mathcal{B}}_j$ denote the samples that arrive in $(s_{j-1}^\textsc{WTA},s_j^\textsc{WTA}]$
and are processed by \textsc{opt} in the same interval.
Let $\mathcal{L}_j$ denote the samples that arrive in $(s_{j-1}^\textsc{WTA},s_j^\textsc{WTA}]$
but are not processed by $\textsc{opt}$ in the same interval.
Define $\mathcal{L}_0:=\emptyset$.
Note that $\tilde{\mathcal{B}}_j\cup\mathcal{L}_j=\mathcal{B}_j^\textsc{WTA}$,
and $\mathcal{L}_{m^\textsc{WTA}}=\emptyset$ (because the final batch of \textsc{opt}
would be processed at $t=a_n$, before the final batch of \textsc{WTA}).
Finally, let $\text{PNE}(j)$ denote the previous batch before $j$ where $\mathcal{S}_j$
is non-empty (the ``previous non-empty'' batch), i.e., $\text{PNE}(j):=\max\{k:k<j\ \text{and}\ \mathcal{S}_k\neq\emptyset\}$.
Define $\text{PNE}(1):=0$.


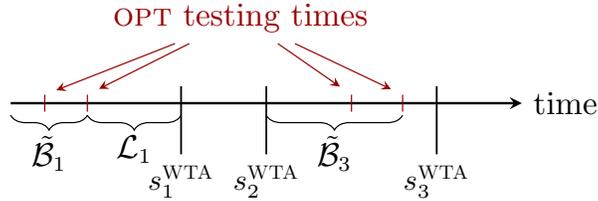
\begin{figure}[h!]
\centering
\resizebox{0.5\columnwidth}{!}{
    \begin{tikzpicture}[>=stealth]
        \draw[thick,->] (-1,0) -- (5,0);
        \node at (5,0) [anchor=west] {time};
        \draw[semithick] (1,0)+(0,0.2) -- +(0,-0.6);
        \node at (1,-0.6) [anchor=north] {\small $s_1^\textsc{WTA}$};
        \draw[semithick] (2,0)+(0,0.2) -- +(0,-0.6);
        \node at (2,-0.6) [anchor=north] {\small $s_2^\textsc{WTA}$};
        \draw[semithick] (4,0)+(0,0.2) -- +(0,-0.6);
        \node at (4,-0.6) [anchor=north] {\small $s_3^\textsc{WTA}$};

        \begin{scope}[color=red!60!black]
        \draw[thin] (-0.1,0)+(0,0.1) node (p1) {} -- +(0,-0.1);
        \draw[thin] (-0.6,0)+(0,0.1) node (p2) {} -- +(0,-0.1);
        \draw[thin] (3.0,0)+(0,0.1) node (p3) {} -- +(0,-0.1);
        \draw[thin] (3.6,0)+(0,0.1) node (p4) {} -- +(0,-0.1);

        \node at (1.7,0.7) [anchor=south] (p) {\textsc{opt} testing times};
        \draw[->] (p) -- (p1);
        \draw[->] (p) -- (p2);
        \draw[->] (p) -- (p3);
        \draw[->] (p) -- (p4);
        \end{scope}

        \draw[decorate, decoration={brace,amplitude=5pt,raise=4pt,mirror}] (-1,0) -- (-0.1,0) node[midway, anchor=north,yshift=-7pt] {$\tilde{\mathcal{B}}_1$};
        \draw[decorate, decoration={brace,amplitude=5pt,raise=4pt,mirror}] (-0.1,0) -- (1,0) node[midway, anchor=north,yshift=-7pt] {$\mathcal{L}_1$};
        \draw[decorate, decoration={brace,amplitude=5pt,raise=4pt,mirror}] (2,0) -- (3.6,0) node[midway, anchor=north,yshift=-7pt] {$\tilde{\mathcal{B}}_3$};
    \end{tikzpicture}
}
\caption{An illustration of the variables used in the proof of Lemma~\ref{claim:testing-cost}.
The calligraphic variables $\mathcal{B}$ and $\mathcal{L}$ denote the set of
samples that arrive during the indicated period.}
\label{fig:proof-variables}
\end{figure}

\begin{table}[h!]
\centering
    \begin{tabular}{c|p{11cm}}
    \hline
    \textbf{Variable} & \textbf{Definition} \\
    \hline
    $\Gamma$ & $\inf_{\mathcal{X},\mathcal{Y}\subseteq\mathcal{F}}\frac{f(\mathcal{X}\cup\mathcal{Y})}{f(\mathcal{X})+f(\mathcal{Y})}$ \\
    \hline
    \textsc{opt} & An optimal schedule (computed with hindsight). \\
    \hline
    \rule{0pt}{2.5ex}$\tilde{\mathcal{B}}_j$ & Samples processed in \textsc{opt} before being processed as part of $\mathcal{B}_j^\textsc{WTA}$. Note that $\tilde{\mathcal{B}}_j\subseteq\mathcal{B}_j^\textsc{WTA}$. \\
    \hline
    \rule{0pt}{2.5ex}$\mathcal{L}_j$ & Samples processed in $\mathcal{B}_j^\textsc{WTA}$ which arrive after \textsc{opt} has processed $\tilde{\mathcal{B}}_j$.
    Note that $\tilde{\mathcal{B}}_j\cup\mathcal{L}_j=\mathcal{B}_j^\textsc{WTA}$. Also $\mathcal{L}_{m^\textsc{WTA}}=\emptyset$ and let $\mathcal{L}_0:=\emptyset$.\\
    \hline
    $\mathcal{S}_j$ & Set of batches that are processed in \textsc{opt} at times $t\in(s_{j-1}^\textsc{WTA},s_j^\textsc{WTA}]$. \\
    \hline
    $\text{PNE}(j)$ & The ``previous non-empty'' batch before $j$ where $\mathcal{S}_j\neq\emptyset$.
    Specifically, $\text{PNE}(j):=\max\{k:k<j\text{ and }\mathcal{S}_k\neq\emptyset\}$. Let $\text{PNE}(1):=0$.\\
    \hline
    \end{tabular}
    \begin{tabular}{c}\rule{0pt}{1ex}\end{tabular}
\caption{Variables used in the proof of Lemma~\ref{claim:testing-cost}.}
\label{tab:proof-variables}
\end{table}

Since $\mathcal{B}_j^\textsc{WTA}=\tilde{\mathcal{B}}_j\cup\mathcal{L}_j$,
Assumption~\ref{assm:f-assumptions} gives
\begin{align}
    f\left( \mathcal{B}_j^\textsc{WTA}\right) \le f(\tilde{\mathcal{B}}_j) + f\left(\mathcal{L}_j\right).
\label{eq:amortize-part1}
\end{align}
Further, by Assumption~\ref{assm:consecutive-sets}, \textsc{opt} needs to process all samples in $\tilde{\mathcal{B}}_j$
and $\mathcal{L}_{\text{PNE}(j)}$ during $(s_{j-1}^\textsc{WTA},s_j^\textsc{WTA}]$, either together or separately,
possibly with other samples if $\text{PNE}(j)<j-1$.
Together with Assumption~\ref{assm:f-assumptions}, this gives
\begin{align}
    f(\tilde{\mathcal{B}}_j\cup\mathcal{L}_{\text{PNE}(j)}) \le \sum_{k\in\mathcal{S}_j}f(\mathcal{B}_k^\textsc{opt}).
\label{eq:amortize-part2}
\end{align}

Adding \eqref{eq:amortize-part1} and \eqref{eq:amortize-part2}, we get
\begin{align*}
    f(\mathcal{B}_j^\textsc{WTA}) \le \sum_{k\in\mathcal{S}_j}f(\mathcal{B}_k^\textsc{opt})+ f(\tilde{\mathcal{B}}_j) + f(\mathcal{L}_j) - f(\tilde{\mathcal{B}}_j\cup\mathcal{L}_{\text{PNE}(j)}).
\end{align*}
Adding and subtracting $f(\mathcal{L}_{\text{PNE}(j)})$ gives
{\begin{align*}
    f(\mathcal{B}_j^\textsc{WTA}) &\!\le\! \sum_{k\in\mathcal{S}_j}\!\!\!f(\!\mathcal{B}_k^\textsc{opt}\!)
        \!+\! f(\!\tilde{\mathcal{B}}_j\!) \!+\! f(\!\mathcal{L}_{\text{PNE}(j)}\!) \!-\! f(\!\tilde{\mathcal{B}}_j\!\cup\!\mathcal{L}_{\text{PNE}(j)}\!)
        \!+\!f(\mathcal{L}_j)\!-\!f(\mathcal{L}_{\text{PNE}(j)}) \\
        &\!\le\! \sum_{k\in\mathcal{S}_j}\!\!\!f(\!\mathcal{B}_k^\textsc{opt}\!)
        \!+\! f(\!\tilde{\mathcal{B}}_j\!\cup\!\mathcal{L}_{\text{PNE}(j)}\!)
            \!\left(\!\tfrac{f(\tilde{\mathcal{B}}_j) \!+\! f(\mathcal{L}_{\text{PNE}(j)})}{f(\tilde{\mathcal{B}}_j\cup\mathcal{L}_{\text{PNE}(j)})}-1\!\right)\!
        \!+\!f(\!\mathcal{L}_j\!)\!-\!f(\!\mathcal{L}_{\text{PNE}(j)}\!).
\end{align*}}
Using the definition of $\Gamma$ and substituting \eqref{eq:amortize-part2}, we obtain
{\begin{align*}
    f(\mathcal{B}_j^\textsc{WTA}) &\le \sum_{k\in\mathcal{S}_j}f(\mathcal{B}_k^\textsc{opt})
        +\left(\sum_{k\in\mathcal{S}_j}f(\mathcal{B}_k^\textsc{opt})\right)\left(\frac{1}{\Gamma}-1\right)
        + f(\mathcal{L}_j)-f(\mathcal{L}_{\text{PNE}(j)}),
\end{align*}}
which simplifies to
\begin{align}
    f(\mathcal{B}_j^\textsc{WTA}) \le \frac{1}{\Gamma}\sum_{k\in\mathcal{S}_j}f(\mathcal{B}_k^\textsc{opt})
        +f(\mathcal{L}_j)-f(\mathcal{L}_{\text{PNE}(j)}).
\label{eq:amortize-part3}
\end{align}

Substituting \eqref{eq:amortize-part3} in \eqref{eq:testing-cost-non-empty-partitions}, we get
{\begin{align*}
    \int_{s^\textsc{WTA}_{j-1}}^{s^\textsc{WTA}_j}\left(u^\textsc{WTA}_\tau-u^\textsc{opt}_\tau\right)^+d\tau
    \le \frac{\alpha}{\Gamma}\sum_{k\in\mathcal{S}_j}f(\mathcal{B}_k^\textsc{opt})
    +\alpha\left(f(\mathcal{L}_j)-f(\mathcal{L}_{\text{PNE}(j)})\right).
\end{align*}}
Summing this over all non-empty $\mathcal{S}_j$ results in a telescopic cancellation of
the $f(\mathcal{L}_j)-f(\mathcal{L}_{\text{PNE}(j)})$ terms.
Using $\mathcal{L}_0=\mathcal{L}_{m^\textsc{WTA}}=\emptyset$, together with
\eqref{eq:partitioned-optimal-cost-empty} and \eqref{eq:proof-claim-1-eq1}, this gives us
\begin{align*}
    \int_0^\infty (u^\textsc{WTA}_\tau-u^\textsc{opt}_\tau)^+d\tau
    \le \frac{\alpha}{\Gamma}\left(nF^\textsc{opt}\right),
\end{align*}
which completes the proof.\hfill$\square$

\begin{remark}
\isit{One can view relation \eqref{eq:amortize-part3} as a variant of the amortized local-competitiveness argument \citep{Edmonds2000, Pruhs2007}.
If the right side only had the first term (possibly with additional factors), it would be locally competitive.}
However, using the $f(\mathcal{L}_j)-f(\mathcal{L}_{\text{PNE}(j)})$ term,
we amortize this cost over all $j$.
This can also be interpreted as a (discrete) potential function \citep{BansalPS2010,VazeN2022} that sums to $0$ over all time.
\end{remark}

\section{A Lower Bound for the Competitive Ratio}
\label{sec:lower-bound}

\revision{We now derive a lower bound on the competitive ratio and state it formally as Theorem~\ref{thm:lower-bound}.}
\begin{theorem}
\revision{Let \textsc{alg} be an online algorithm that achieves a competitive ratio of \(\rho\),  and \(f(\cdot)\) be the
processing-cost function for the batching problem.} 
Then,
\begin{align*}
\rho \ge \frac{1}{{\Gamma}},
\end{align*}
where \(\Gamma=\inf_{\mathcal{X},\mathcal{Y}\subseteq\mathcal{F}}\frac{f(\mathcal{X}\cup\mathcal{Y})}{f(\mathcal{X})+f(\mathcal{Y})}\).
\label{thm:lower-bound}
\end{theorem}
\begin{proof}
For proving a lower bound given any algorithm, we wish to construct a maximally adversarial sequence of arrivals, which
makes the algorithm's cost as high as possible compared to the optimal schedule.
Let a set of arrivals with feature vectors $\mathcal{X}_1$ occur at time $t=0(=:t_0)$ (we later choose  $\mathcal{X}_1$ to get the maximum lower bound),
and assume \textsc{alg} processes them at time $t=t_1$.
We construct the arrival sequence so that further arrivals with a feature-vector set $\mathcal{X}_2$ occur at $t=t_1+\epsilon$ for some small $\epsilon>0$.
Let \textsc{alg} process these at some time $t_2\ge t_1+\epsilon$.
At time $t=t_2+\epsilon$, we again add a set of arrivals with feature vectors $\mathcal{X}_1$, with the set $\mathcal{X}_2$ subsequently immediately after the previous set $\mathcal{X}_1$ is processed.
We repeat this adversarial construction $s$ times so that we have $2s$ batches of arrivals, where each batch either has feature vectors $\mathcal{X}_1$ or $\mathcal{X}_2$.
The $j$th batch
arrives at time $t=t_{j-1}+\epsilon$ and is processed at $t=t_j$.%
\footnote{We can assume the first batch arrives at $t=\epsilon$ rather than $t=0$ with no meaningful difference to the analysis.}
This gives us the unnormalized cost of \textsc{alg}, \revision{$s\times(|\mathcal{X}_1|+|\mathcal{X}_2|)\times J^\textsc{alg}$},
for this arrival sequence as
{\begin{align}
    s(|\mathcal{X}_1|\!+\!|\mathcal{X}_2|)J^\textsc{alg} &= s(f(\mathcal{X}_1)\!+\!f(\mathcal{X}_2))
    + \sum_{k=1}^s\! { |\mathcal{X}_1|(t_{2k-1}\!-\!t_{2k-2}\!-\!\epsilon)}
    + \sum_{k=1}^s\! { |\mathcal{X}_2|(t_{2k}\!-\!t_{2k-1}\!-\!\epsilon)}\nonumber\\
    = \!s(f(\mathcal{X}_1)\!+\!f(\mathcal{X}_2))\!
    &+\! \sum_{k=1}^s\!\!\Big(\! { |\mathcal{X}_1|(t_{2k-1}-t_{2k-2})+|\mathcal{X}_2|(t_{2k}-t_{2k-1})}\!\Big)
    \!-\! (|\mathcal{X}_1|\!+\!|\mathcal{X}_2|)s\epsilon.
    \label{eq:jalg-lower-bound-eq}
\end{align}}
We now need to compare $J^\textsc{alg}$ with the optimal cost $J^\textsc{opt}$ for these arrivals to
determine the competitive ratio.
While we can use the techniques of Sec.~\ref{sec:offline} to compute the optimal schedule for a particular set of $\{t_1,t_2,\ldots,t_s\}$,
this involves solving a dynamic programming problem and it does not seem analytically tractable to express the solution as a function of
$\{t_1,t_2,\ldots,t_s\}$.
Instead, adapting a technique by \cite{DoolyGS2001}, we use the fact that the optimal cost can be no more than the average of
the cost of processing at odd-numbered batch arrivals and the cost of processing at even-numbered batch arrivals.

The odd-numbered batches arrive at $t_0, t_2+\epsilon, t_4+\epsilon,\ldots,t_{2s-2}+\epsilon$.
If we process at these time instants, we also need to process the final batch, which arrives at $t=t_{2s-1}+\epsilon$.
To simplify the math, let us assume the ``odd'' schedule processes the final batch at time $t_{2s}+\epsilon$.
With this setup, the first batch has feature vectors $\mathcal{X}_1$ and $2s$th arrivals have the feature vectors $\mathcal{X}_2$,
and the remaining $(2s-2)$ arrivals get processed
in $(s-1)$ batches of feature vectors  $\mathcal{X}_1\cup\mathcal{X}_2$.
Further, the odd-numbered arrivals do not experience any waiting time, while the even-numbered arrivals experience 
a waiting time of $t_j-t_{j-1}$ for batch $j$.
\revision{This gives us the cost of this schedule, $J^\textsc{odd}$ as}
\begin{align}
    \!\!\!\!\!\!\!s(|\mathcal{X}_1|\!+\!|\mathcal{X}_2|)J^\textsc{odd}\! =\! f(\mathcal{X}_1)\!+\!f(\mathcal{X}_2)\! + \!(s\!-\!1)f(\mathcal{X}_1\!\cup\!\mathcal{X}_2)\! +\! \sum_{k=1}^s\!|\mathcal{X}_2|(t_{2k}\!-\!t_{2k-1}).
    \label{eq:jodd-lower-bound-eq}
\end{align}

Similarly, processing at each even-numbered arrival implies processing $s$ batches each with feature vectors $\mathcal{X}_1\cup\mathcal{X}_2$.
The waiting time here would be for the odd-numbered arrivals, equal to $(t_j-t_{j-1})$ for batch $j$.
\revision{This gives}
\begin{align}
    s(|\mathcal{X}_1|+|\mathcal{X}_2|)J^\textsc{even} = sf(\mathcal{X}_1\cup\mathcal{X}_2) + \sum_{k=1}^s|\mathcal{X}_1|(t_{2k-1}-t_{2k-2}).
    \label{eq:jeven-lower-bound-eq}
\end{align}
Averaging \eqref{eq:jodd-lower-bound-eq} and \eqref{eq:jeven-lower-bound-eq}, we get
{\begin{align} 
    &s(|\mathcal{X}_1|+|\mathcal{X}_2|)J^\textsc{opt} \le \frac{s(|\mathcal{X}_1|+|\mathcal{X}_2|)(J^\textsc{odd}+J^\textsc{even})}{2} \nonumber\\
    &= \frac{f(\mathcal{X}_1)+f(\mathcal{X}_2)}{2} + \frac{2s-1}{2}f(\mathcal{X}_1\!\cup\!\mathcal{X}_2)
    + \frac{1}{2}\!\sum_{k=1}^s\!\!\Big({ |\mathcal{X}_1|(t_{2k-1}-t_{2k-2})+|\mathcal{X}_2|(t_{2k}-t_{2k-1})}\Big).
    \label{eq:jopt-average-lower-bound}
\end{align}}
Using \eqref{eq:jalg-lower-bound-eq} and \eqref{eq:jopt-average-lower-bound} with the fact that
$\epsilon$ can be arbitrarily small%
\footnote{\thirdrevision{Recall that we are constructing the arrival times here, so $\epsilon$ can be any positive value, however small.}} 
gives
{\small
\begin{align*}
    \frac{J^\textsc{alg}}{J^\textsc{opt}} \ge \frac{s(f(\mathcal{X}_1)+f(\mathcal{X}_2))+\sum_{k=1}^s\Big(|\mathcal{X}_1|(t_{2k-1}-t_{2k-2})+|\mathcal{X}_2|(t_{2k}-t_{2k-1})\Big)}
        {\frac{f(\mathcal{X}_1)+f(\mathcal{X}_2)}{2}+\frac{2s-1}{2}f(\mathcal{X}_1\cup\mathcal{X}_2)+\tfrac{1}{2}\sum_{k=1}^s\Big(|\mathcal{X}_1|(t_{2k-1}-t_{2k-2})+|\mathcal{X}_2|(t_{2k}-t_{2k-1})\Big)}.
\end{align*}}
Rearranging the terms, this gives
{\small
\begin{align*}
    \frac{J^\textsc{alg}}{J^\textsc{opt}} \ge 
    2\left(\frac{f(\mathcal{X}_1)+f(\mathcal{X}_2)+\frac{1}{s}\sum_{k=1}^s\Big(|\mathcal{X}_1|(t_{2k-1}-t_{2k-2})+|\mathcal{X}_2|(t_{2k}-t_{2k-1})\Big)}{2f(\mathcal{X}_1\cup\mathcal{X}_2)+\frac{f(\mathcal{X}_1)+f(\mathcal{X}_2)-f(\mathcal{X}_1\cup\mathcal{X}_2)}{s}+\frac{1}{s}\sum_{k=1}^s\Big(|\mathcal{X}_1|(t_{2k-1}-t_{2k-2})+|\mathcal{X}_2|(t_{2k}-t_{2k-1})\Big)}\right).
\end{align*}}
For $0< a\le b$ and $c\ge 0$, we have $\frac{a+c}{b+c}\ge\frac{a}{b}$.
So we can remove the $\sum_{j=1}^s(t_j-t_{j-1})$ terms to get
\begin{align*}
    \frac{J^\textsc{alg}}{J^\textsc{opt}} \ge 
    2\left(\frac{f(\mathcal{X}_1)+f(\mathcal{X}_2)}{2f(\mathcal{X}_1\cup\mathcal{X}_2)+\frac{f(\mathcal{X}_1)+f(\mathcal{X}_2)-f(\mathcal{X}_1\cup\mathcal{X}_2)}{s}}\right).
\end{align*}
For any given $\mathcal{X}_1,\mathcal{X}_2$, we can make $s$ arbitrarily large to get the limiting bound
\begin{align*} 
    \frac{J^\textsc{alg}}{J^\textsc{opt}} \ge 
    \left(\frac{f(\mathcal{X}_1)+f(\mathcal{X}_2)}{f(\mathcal{X}_1\cup\mathcal{X}_2)}\right),
\end{align*}
and since we can also choose $\mathcal{X}_1,\mathcal{X}_2$, we get the lower bound
\begin{align*}
    \frac{J^\textsc{alg}}{J^\textsc{opt}} \ge \frac{1}{\Gamma},
\end{align*}
where $\Gamma=\inf_{\mathcal{X},\mathcal{Y}\subseteq\mathcal{F}}\frac{f(\mathcal{X}\cup\mathcal{Y})}{f(\mathcal{X})+f(\mathcal{Y})}$. This completes the proof.
\end{proof}

\citet{BhimarajuV2022} have a lower bound for a special case where the samples are all homogenous, i.e.,
there exists a function $g:\mathbb{N}\cup\{0\}\mapsto\mathbb{R}^+\cup\{0\}$ such that $f(\mathcal{X})=g(|\mathcal{X}|)$.
Using our definition of $\Gamma$, their lower bound gives
$1+\left(\sqrt{\Gamma^2-2\Gamma+2}-\Gamma\right)$, for all feasible values of $\Gamma$ (i.e., $\Gamma\in\left[\frac{1}{2},1\right]$).
See Fig.~\ref{fig:bound-comparison-text} for a plot with both the bounds.
\begin{figure*}[h]
    \begin{subfigure}[t]{0.5\textwidth}
        \centering
        \includegraphics[width=\textwidth]{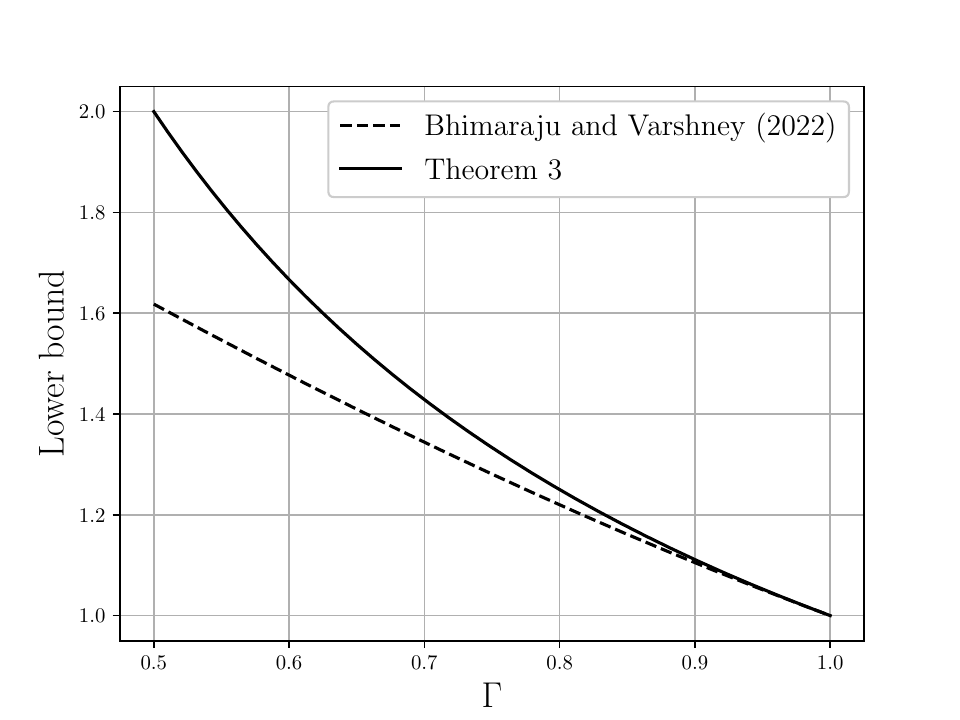}
        \caption{}
        \label{fig:bound-comparison}
    \end{subfigure}%
    \begin{subfigure}[t]{0.5\textwidth}
        \centering
        \includegraphics[width=\textwidth]{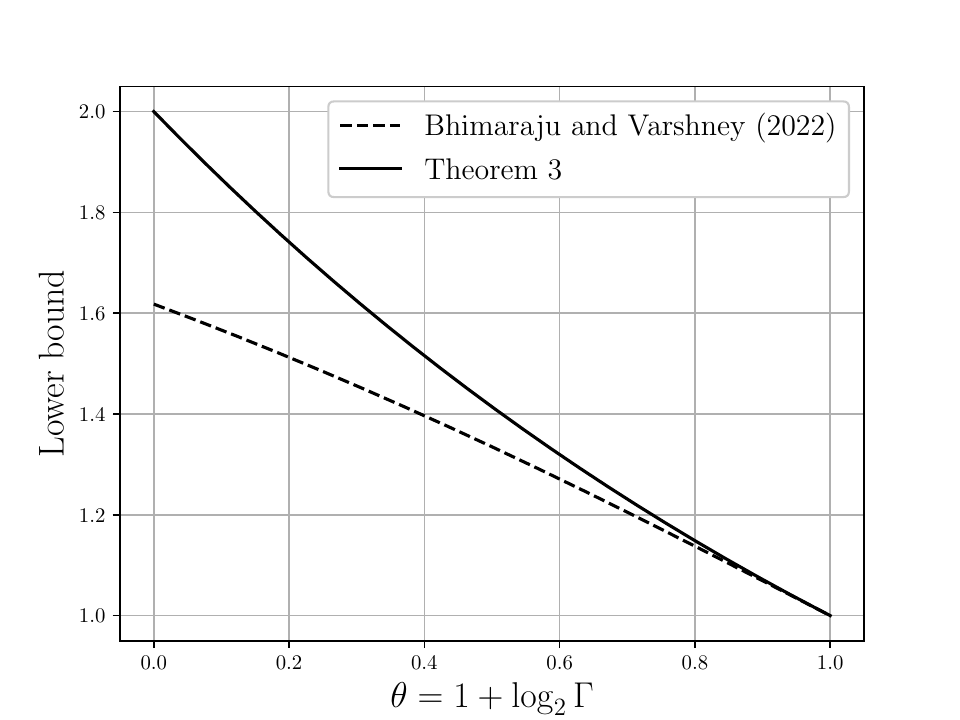}
        \caption{}
        \label{fig:bound-comparison2}
    \end{subfigure}
    \caption{Comparison of the bound in Theorem~\ref{thm:lower-bound} with the result in \citet{BhimarajuV2022}. For the setting in \citet{BhimarajuV2022}, the processing cost is just a function of the number of processed elements, and we have $\Gamma=\inf_{x\in\mathbb{N}}\frac{f(2x)}{2f(x)}\in\left[\frac{1}{2},1\right]$ for
    concave functions using our definition of $\Gamma$. We can see that Theorem~\ref{thm:lower-bound} is stronger for all values of $\Gamma$ in this domain.
    In Fig.~\ref{fig:bound-comparison}, we plot the two bounds as a function of $\Gamma$.
    In Fig.~\ref{fig:bound-comparison2}, we consider the processing-cost functions $f(x)=x^\theta$ for $\theta\in[0,1]$.
    This gives us $\Gamma=2^{\theta-1}$, and we plot the two bounds as a function of $\theta$.
    }
    \label{fig:bound-comparison-text}
\end{figure*}

We also note that Theorem~\ref{thm:lower-bound} is a generalization of the lower bound in \citet{DoolyGS2001}.
The setting in \citet{DoolyGS2001} is equivalent to fixing the processing-cost function to be equal to a constant.
In this case, we have $\Gamma=\frac{1}{2}$, and the lower bound we get is $\frac{1}{\Gamma}=2$, which matches the 
bound of \citet{DoolyGS2001}.
However, Theorem~\ref{thm:lower-bound} is applicable to any function which satisfies Assumption~\ref{assm:f-assumptions},
and is thus much more general.
\section{Numerical Experiments}
\label{sec:simulations}

\secondrevision{
In this section, we empirically evaluate our online algorithms. 
While the exact cost functions may differ substantially by application, our aim here is to investigate performance in a more general sense.
Thus we use cost functions that depend only on the number of arrivals, i.e.,
$f(\mathcal{X})=\tilde{f}(|\mathcal{X}|)$ for some $\tilde{f}$.
Specifically, the following are used depending on the experiment:
(i) $\tilde{f}(n)=\sqrt{n}$, (ii) $\tilde{f}(n)=\log(1+n)$, and (iii) $\tilde{f}(n)=\min\{3n,10\}$.
These span a relatively wide range of $\Gamma$ from $\frac{1}{\sqrt{2}}$ for (i) to $\frac{1}{2}$ for (ii) and (iii).
Further, while $\log(1+n)$ and $\min\{3n,10\}$ both have $\Gamma=\frac{1}{2}$,
$\log(1+n)$ achieves its minimal diminishing costs asymptotically, while $\min\{3n,10\}$ saturates at $n=4$.
We believe these settings cover a wide range of scenarios that might occur in practice.
}

\subsection{\revision{Simulated data}}

\revision{For the simulations, we use both homogeneous and inhomogeneous Poisson processes to generate the arrivals.}
We refer to \cite{Ross2023} for details on how to simulate such arrivals.

\revision{In selecting the simulation experiments to perform, we aim to evaluate the empirical performance of the algorithms under different processing cost functions, arrival rates over time, and total numbers of samples. Fig.~\ref{fig:sqrt-log-cost} varies the number of samples $n$, while Fig.~\ref{fig:sqrt-log-cost-rate} varies the arrival rate $\lambda$. Fig.~\ref{fig:vs-wte} compares the performance of \textsc{WTA} with \textsc{WTE} \citep{BhimarajuV2022}, and Fig.~\ref{fig:alpha} compares the performance of the algorithm for different values of the parameter $\alpha$ in \textsc{WTA}. Finally, Fig.~\ref{fig:changing-lambda} considers the performance of \textsc{WTA} for a time-inhomogeneous arrival rate. We discuss the results in detail below.
}

\begin{figure*}[h]
    \begin{subfigure}[t]{0.5\textwidth}
        \centering
        \includegraphics[width=\textwidth]{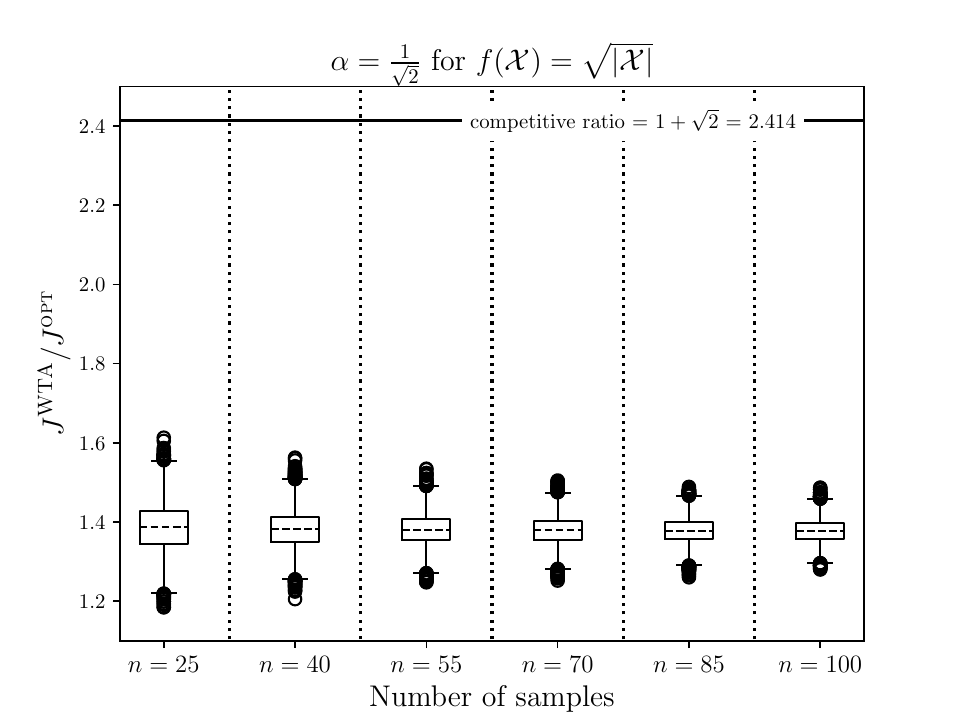}
        \caption{}
        \label{fig:sqrt-cost}
    \end{subfigure}%
    \begin{subfigure}[t]{0.5\textwidth}
        \centering
        \includegraphics[width=\textwidth]{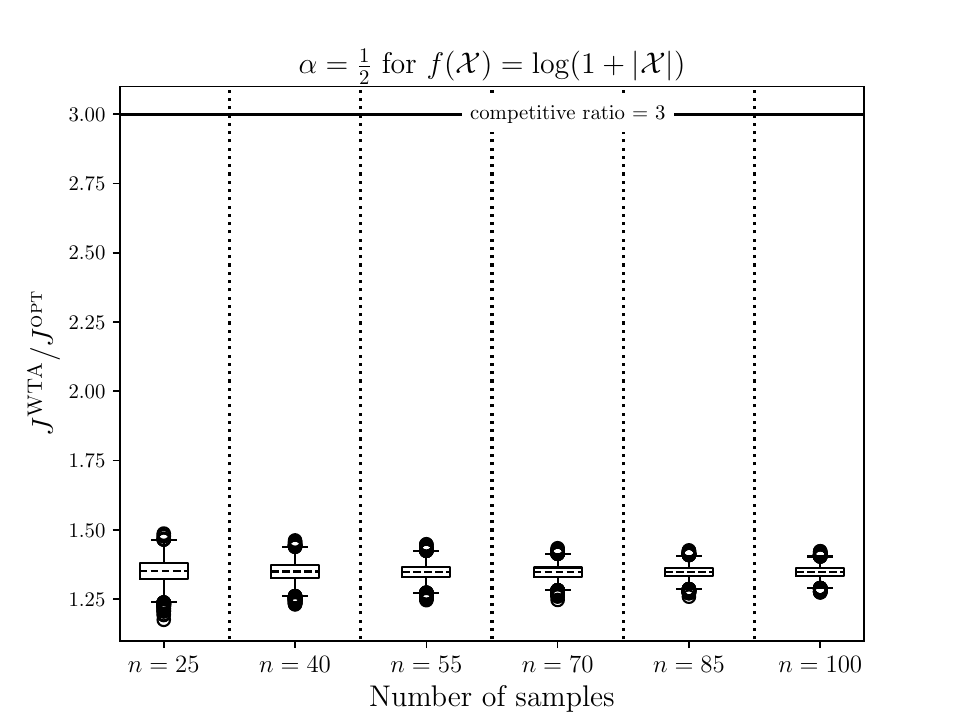}
        \caption{}
        \label{fig:log-cost}
    \end{subfigure}
    \caption{Performance of \textsc{WTA} over a varying number of samples. We use a Poisson process with rate $\lambda=2$ to generate $n$ arrivals. The value of $\alpha$ for \textsc{WTA} is set based on the ``optimal'' value given by Corollary~\ref{cor:wte}(iii). We plot the results over $10000$ trials for each value of $n$.
    }
    \label{fig:sqrt-log-cost}
\end{figure*} 
In Fig.~\ref{fig:sqrt-log-cost}, we show the performance of \textsc{WTA} over a varying
number of arrivals $n$ for two different cost functions: $f(\mathcal{X})=\sqrt{|\mathcal{X}|}$ in Fig.~\ref{fig:sqrt-cost} and $f(\mathcal{X})=\log(1+|\mathcal{X}|)$ in Fig.~\ref{fig:log-cost}.
In both cases, we generate the arrivals using a time-homogeneous Poisson process with a (constant)
rate $\lambda=2$.
For $f(\mathcal{X})=\sqrt{|\mathcal{X}|}$, we have $\Gamma=\inf_{x\in\mathbb{N}}\frac{f(2x)}{2f(x)}=\frac{1}{\sqrt{2}}$, and we use $\alpha=\frac{1}{\sqrt{2}}$ as suggested by Corollary~\ref{cor:wte}(iii).
Likewise, we have $\Gamma=\frac{1}{2}$ for $f(\mathcal{X})=\log(1+|\mathcal{X}|)$ and so we use $\alpha=\frac{1}{2}$.
For each value of $n$, we generate the arrivals $10000$ times and compute $\frac{J^\textsc{WTA}}{J^\textsc{opt}}$
for each instance ($J^\textsc{opt}$ can be calculated using Alg.~\ref{alg:osp}).
The distribution of $\frac{J^\textsc{WTA}}{J^\textsc{opt}}$ is then shown in Fig.~\ref{fig:sqrt-log-cost}.

We can see that in both Fig.~\ref{fig:sqrt-cost} and Fig.~\ref{fig:log-cost}, there is considerable gap
between the maximum value of $\frac{J^\textsc{WTA}}{J^\textsc{opt}}$ and the competitive ratio.
While Corollary~\ref{cor:wte}(iii) guarantees that the value of $\frac{J^\textsc{WTA}}{J^\textsc{opt}}$
would never be more than the competitive ratio, it seems to be the case that a tighter bound might be possible
for \textsc{WTA}.%
\footnote{Of course, there is the possibility that our $10000$ trials missed a very specific instance
where $\frac{J^\textsc{WTA}}{J^\textsc{opt}}$ is very close to the competitive ratio.}
Further, we can also see that the distribution of $\frac{J^\textsc{WTA}}{J^\textsc{opt}}$
concentrates around a certain value as the value of $n$ increases.
Note, however, that this does not reduce the competitive ratio for large values of $n$.
We can repeat an $n=25$ instance $4$ times (with an appropriate gap between the repetitions)
to get an instance with the same ratio of $\frac{J^\textsc{WTA}}{J^\textsc{opt}}$ for $n=100$.
However, such ``repeated'' instances are unlikely to occur when the arrivals follow a Poisson distribution.

\begin{figure*}[h]
    \begin{subfigure}[t]{0.5\textwidth}
        \centering
        \includegraphics[width=\textwidth]{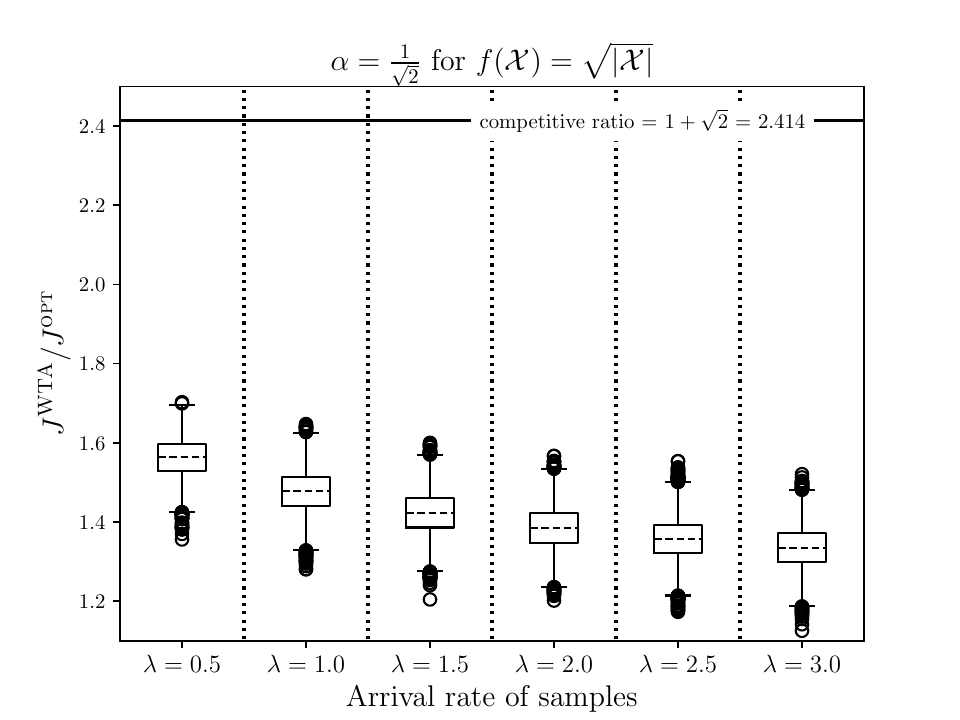}
        \caption{}
        \label{fig:sqrt-cost-rate}
    \end{subfigure}%
    \begin{subfigure}[t]{0.5\textwidth}
        \centering
        \includegraphics[width=\textwidth]{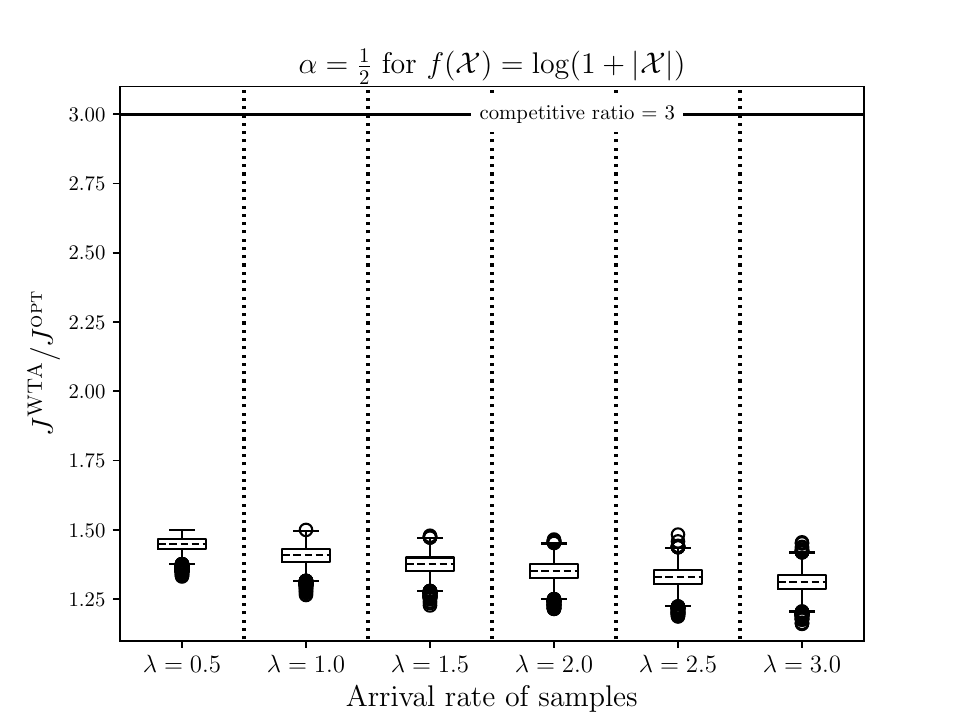}
        \caption{}
        \label{fig:log-cost-rate}
    \end{subfigure}
    \caption{Performance of \textsc{WTA} over a varying arrival rate of samples. We use a Poisson process with rate $\lambda$ to generate $n=30$ arrivals. The value of $\alpha$ for \textsc{WTA} is set based on the ``optimal'' value given by Corollary~\ref{cor:wte}(iii). We plot the results over $10000$ trials for each value of $\lambda$.
    }
    \label{fig:sqrt-log-cost-rate}
\end{figure*}
In Fig.~\ref{fig:sqrt-log-cost-rate}, we repeat the experiment of Fig.~\ref{fig:sqrt-log-cost},
but we fix the value of $n$ to $30$, and instead vary the value of $\lambda$.
We observe that the distribution of $\frac{J^\textsc{WTA}}{J^\textsc{opt}}$ trends downward
as $\lambda$ increases in both Fig.~\ref{fig:sqrt-cost-rate} and Fig.~\ref{fig:log-cost-rate}.
This is intuitively reasonable since a very high rate suggests that waiting for samples is better
than testing too soon since a new sample is probably just around the corner (due to the high rate).
Since \textsc{WTA} waits sufficiently long to give a constant competitive ratio for any possible
set of arrival times, the algorithm performs very well when the rate is high.
However, since the arrivals are stochastic and any possible set of arrival times have a non-zero
probability density, the ``competitive ratio'' is the same in all these cases.

\begin{figure*}[h]
    \begin{subfigure}[t]{0.5\textwidth}
        \centering
        \includegraphics[width=\textwidth]{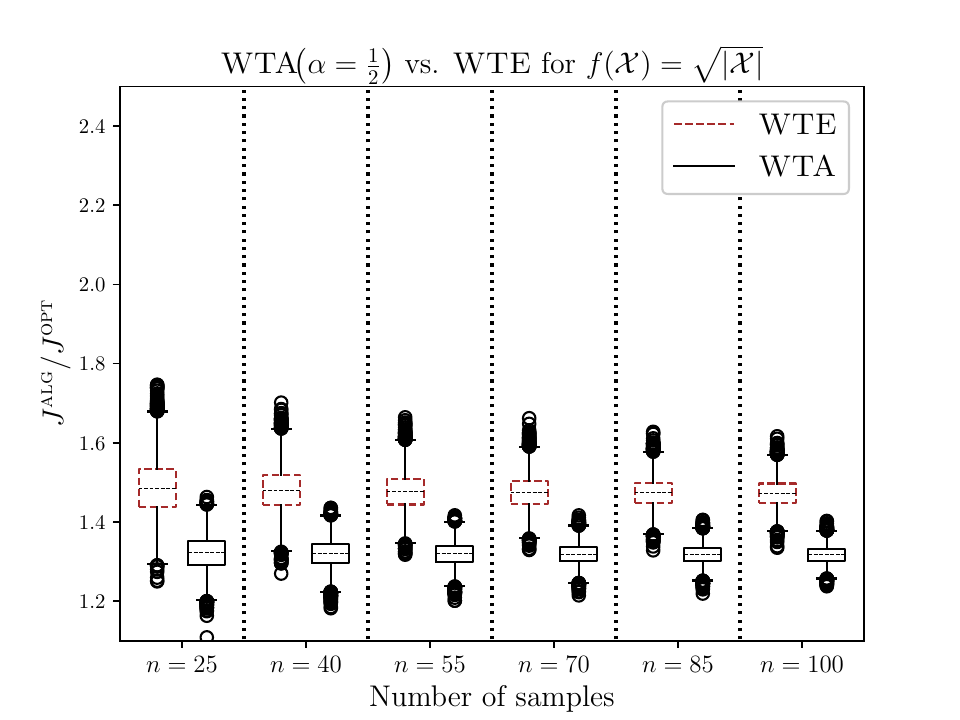}
        \caption{}
        \label{fig:vs-wte-n}
    \end{subfigure}%
    \begin{subfigure}[t]{0.5\textwidth}
        \centering
        \includegraphics[width=\textwidth]{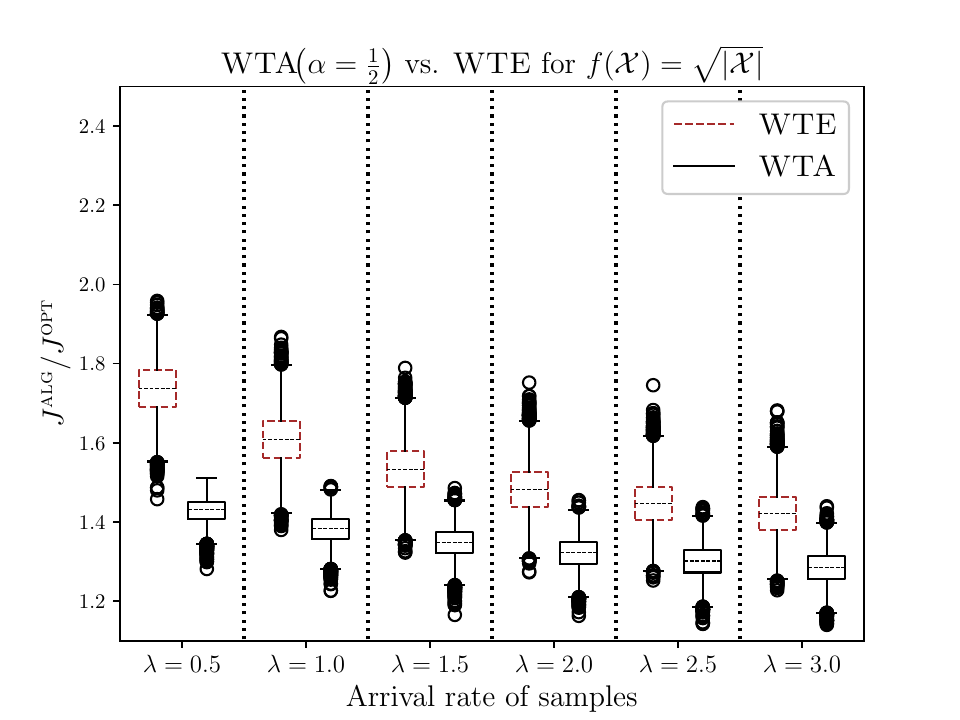}
        \caption{}
        \label{fig:vs-wte-rate}
    \end{subfigure}
    \caption{Comparison of the performance of \textsc{WTA} with $\alpha=\frac{1}{2}$ against \textsc{WTE} \citep{BhimarajuV2022}.
    We use $\lambda=2$ for Fig.~\ref{fig:vs-wte-n} and $n=30$ for Fig.~\ref{fig:vs-wte-rate},
    and plot the results over $10000$ trials for each value of $n$ and $\lambda$.
    }
    \label{fig:vs-wte}
\end{figure*}

In Fig.~\ref{fig:vs-wte}, we compare the performance of \textsc{WTA} (with $\alpha=\frac{1}{2})$ against
\textsc{WTE} \citep{BhimarajuV2022}.
We vary the value of $n$ in Fig.~\ref{fig:vs-wte-n} and the rate $\lambda$ in Fig.~\ref{fig:vs-wte-rate}.
The cost function is $f(\mathcal{X})=\sqrt{|\mathcal{X}|}$ in both cases.
Recall that \textsc{WTA} with $\alpha=1$ reduces to \textsc{WTE}.
We choose $\alpha=\frac{1}{2}$ because that is the function-agnostic value suggested by Corollary~\ref{cor:wte}(i).
Additionally, as we see subsequently in Fig.~\ref{fig:alpha}, this choice turns out to be better experimentally
than $\alpha=\frac{1}{\sqrt{2}}$, which has the best (theoretical) competitive ratio.
Fig.~\ref{fig:vs-wte-n} and Fig.~\ref{fig:vs-wte-rate} show that $\textsc{WTA}\left(\alpha=\frac{1}{2}\right)$ outperforms
\textsc{WTE} for all values of $n$ and $\lambda$.
As we see in Fig.~\ref{fig:vs-wte-rate}, the gains are particularly large for small values of $\lambda$,
where the distribution of $\frac{J^\textsc{WTA}}{J^\textsc{opt}}$
is much lower than the distribution of $\frac{J^\textsc{WTE}}{J^\textsc{opt}}$.
Fig.~\ref{fig:alpha-sqrt} suggests that $\textsc{WTA}\big(\alpha=\frac{1}{\sqrt{2}}\big)$ is likely to outperform
\textsc{WTE} as well, but by a smaller margin than $\textsc{WTA}\left(\alpha=\frac{1}{2}\right)$.

\begin{figure*}[h]
    \begin{subfigure}[t]{0.5\textwidth}
        \centering
        \includegraphics[width=\textwidth]{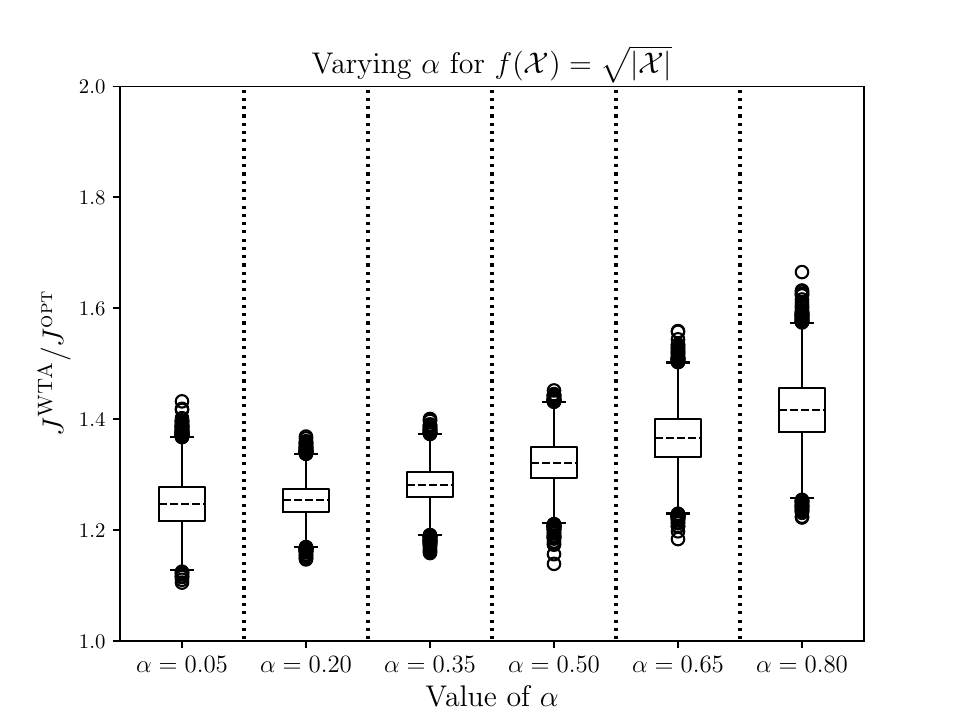}
        \caption{}
        \label{fig:alpha-sqrt}
    \end{subfigure}%
    \begin{subfigure}[t]{0.5\textwidth}
        \centering
        \includegraphics[width=\textwidth]{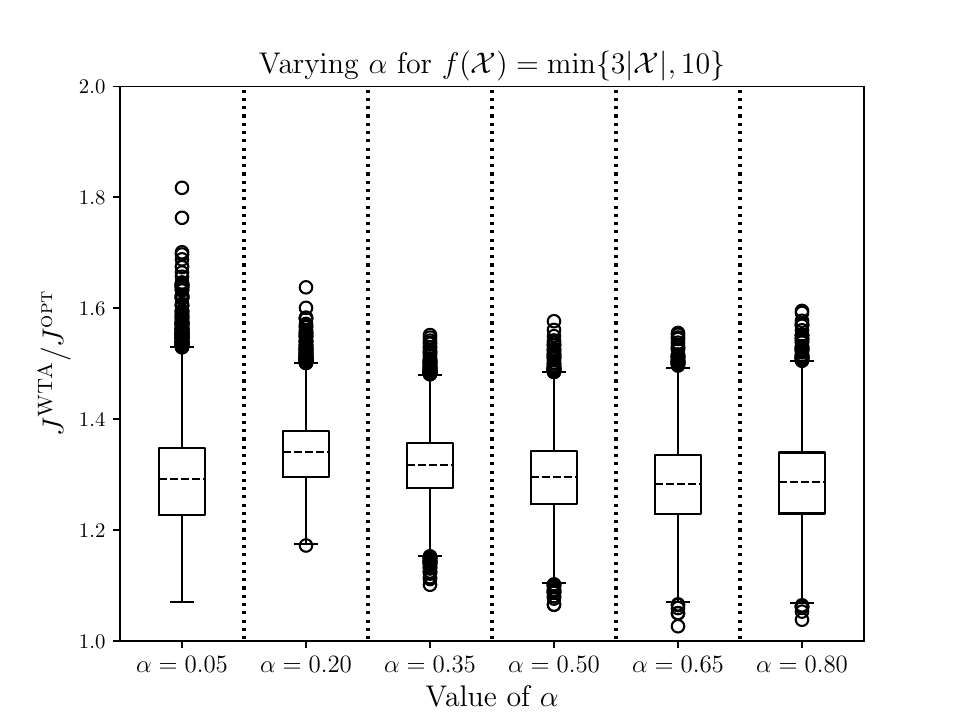}
        \caption{}
        \label{fig:alpha-sqrt-small}
    \end{subfigure}
    \caption{Performance of \textsc{WTA} for different values of $\alpha$. We use $\lambda=2$ and $n=30$ for all values of $\alpha$,
    and plot the results over $10000$ trials.
    }
    \label{fig:alpha}
\end{figure*}

In Fig.~\ref{fig:alpha}, we vary the value of $\alpha$ while fixing $n=30$ and $\lambda=2$.
Fig.~\ref{fig:alpha-sqrt} uses the cost function $f(\mathcal{X})=\sqrt{|\mathcal{X}|}$.
For the worst-case performance, we see in Fig.~\ref{fig:alpha-sqrt} that it improves
from $\alpha=0.05$ to $\alpha=0.2$ while worsening for increasing values of $\alpha$.
However, for some outlier instances, $\alpha=0.05$ performs better.
This is not surprising since if the arrivals are sufficiently spaced apart, the optimal
schedule processes them independently, which corresponds to $\alpha=0$.
We also observe that for $f(\mathcal{X})=\sqrt{|\mathcal{X}|}$, the value of $\alpha$ that has the best (theoretical) competitive
ratio according to Corollary~\ref{cor:wte}, i.e., $\alpha=\frac{1}{\sqrt{2}}$, is not the best practical choice.
While $f(\mathcal{X})=\log(1+|\mathcal{X}|)$ exhibits similar trends, we omit that function and instead use $f(\mathcal{X})=\min\{3|\mathcal{X}|,10\}$
in Fig.~\ref{fig:alpha-sqrt-small}, which has more interesting behavior.
Here, we see that one of the best empirical performances is for $\alpha=\frac{1}{2}$, which also has the
best theoretical guarantee.
We also see that $\alpha=0.65$ performs just as well, perhaps a little better for some outlier instances.

\begin{figure*}[h]
    \begin{subfigure}[t]{0.5\textwidth}
        \centering
        \includegraphics[width=\textwidth]{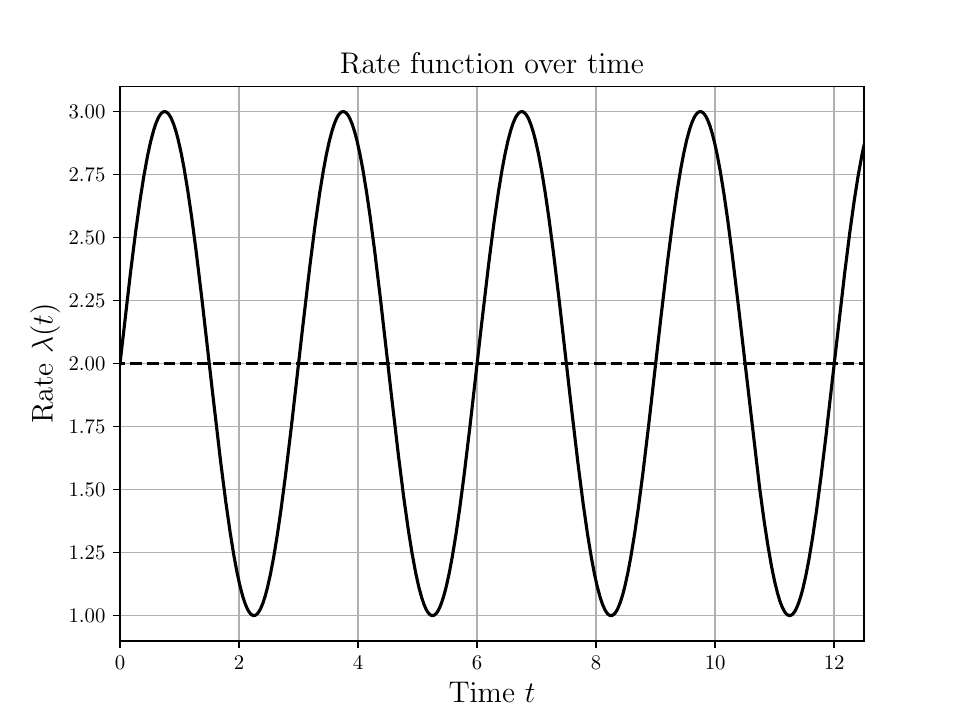}
        \caption{}
        \label{fig:changing-lambda-rate}
    \end{subfigure}%
    \begin{subfigure}[t]{0.5\textwidth}
        \centering
        \includegraphics[width=\textwidth]{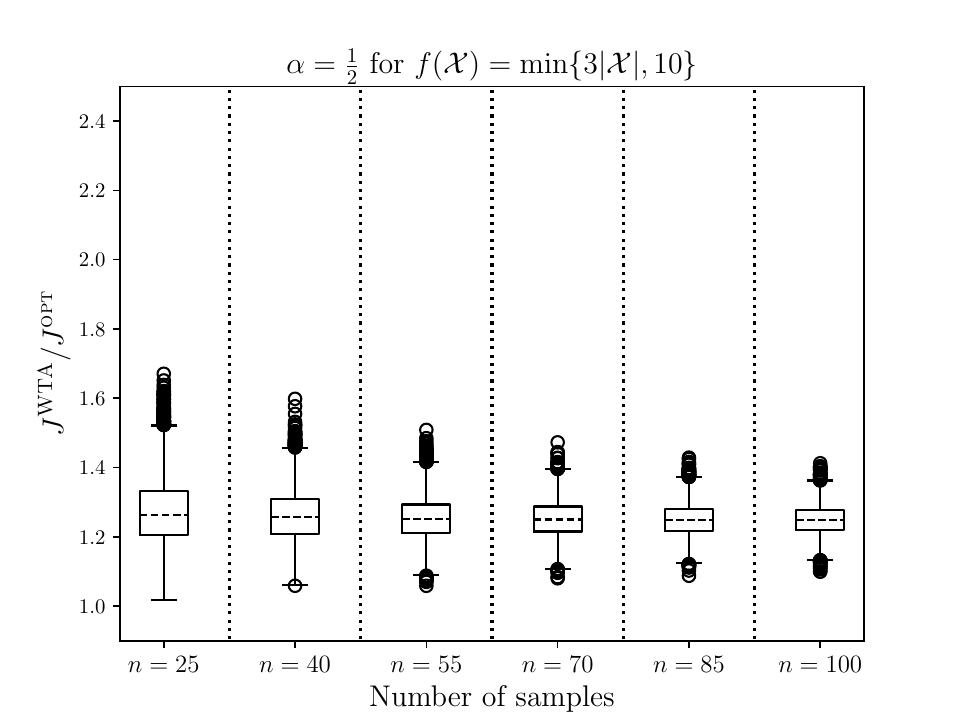}
        \caption{}
        \label{fig:changing-lambda-result}
    \end{subfigure}
    \caption{Performance of \textsc{WTA} for different values of $n$ with a time-inhomogenous arrival rate $\lambda(t)$. We use $\alpha=\frac{1}{2}$ which is the ``optimal'' choice for $f(\mathcal{X})=\min\{3|\mathcal{X}|,10\}$ according to Corollary~\ref{cor:wte}(iii) and plot the results over $ 10000$ trials.
    }
    \label{fig:changing-lambda}
\end{figure*}

While the arrivals so far have been homogeneous (constant-rate) Poisson processes, Fig.~\ref{fig:changing-lambda}
considers a time-inhomogeneous Poisson process.
We use the rate function $\lambda(t)=2+\sin\left(\frac{2\pi}{3}t\right)$ shown in Fig.~\ref{fig:changing-lambda-rate}.
We use the cost function $f(\mathcal{X})=\min\{3|\mathcal{X}|,10\}$ from Fig.~\ref{fig:alpha-sqrt-small}.
Fig.~\ref{fig:changing-lambda-result} shows the performance of \textsc{WTA} with $\alpha=\frac{1}{2}$ for the arrival
process $\lambda(t)$ for a varying $n$.
We see a trend similar to Fig.~\ref{fig:sqrt-log-cost}, where the distribution of $\frac{J^\textsc{WTA}}{J^\textsc{opt}}$
concentrates around a certain value as $n$ increases.
Note that since $\alpha=\frac{1}{2}$ and $\Gamma=\frac{1}{2}$, the competitive ratio guaranteed by Corollary~\ref{cor:wte} is $3$.

\subsection{Real-world data}

\begin{figure*}[h]
    \begin{subfigure}[t]{0.5\textwidth}
        \centering
        \includegraphics[width=\textwidth]{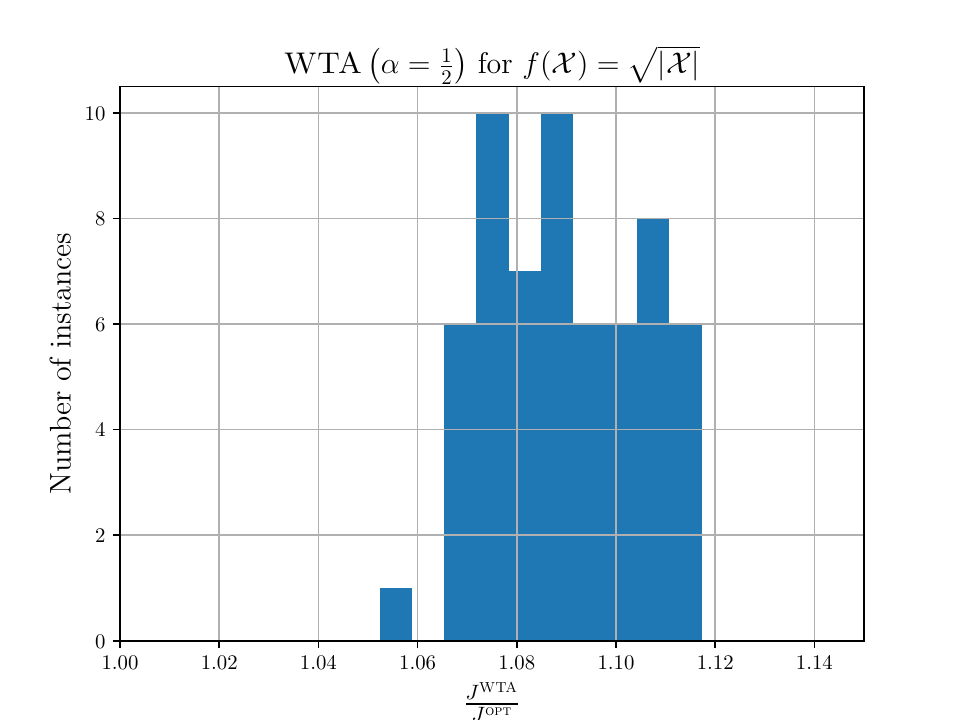}
        \caption{}
        \label{fig:sqrt-real-data}
    \end{subfigure}%
    \begin{subfigure}[t]{0.5\textwidth}
        \centering
        \includegraphics[width=\textwidth]{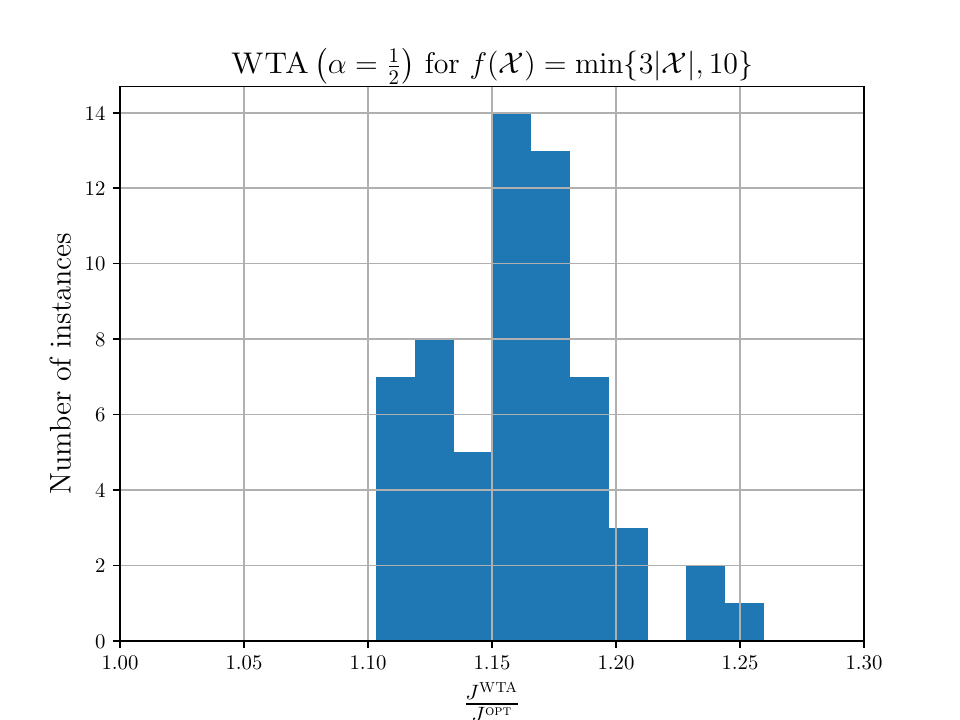}
        \caption{}
        \label{fig:min-real-data}
    \end{subfigure}
    \caption{Performance of \textsc{WTA} with $\alpha=\frac{1}{2}$ over the arrival data from \citet{BishopOOOO2018}.
    In Fig.~\ref{fig:sqrt-real-data} we use the processing cost function $f(\mathcal{X})=\sqrt{|\mathcal{X}|}$ and in Fig.~\ref{fig:min-real-data}
    we use the function $f(\mathcal{X})=\min\{3|\mathcal{X}|,10\}$.
    The best schedule is computed using \textsc{OSP} and we plot the histogram of $\frac{J^\textsc{WTA}}{J^\textsc{opt}}$ over the $60$ instances
    in the data.
    }
    \label{fig:real-data}
\end{figure*}

\revision{As real-world data, we use customer arrivals at three banks over four weeks, as collected by \citet{BishopOOOO2018}. We use this publicly available data because we expect some of our target applications to have similar arrival statistics.}
We consider each day of each week as a problem instance, giving us a total of $3\times4\times5=60$ instances
over which we report the statistics.
For each instance, customers arrive from $8$ AM to $3$ PM (i.e., over a period of $7$ hours), and
the average number of arrivals per day was $874.9$.
We can see in Fig.~\ref{fig:real-data} that the schedule given by \textsc{WTA}
is very close to the optimal schedule.
For all the instances considered here,
the cost of the \textsc{WTA} schedule was less than a factor $1.3$ away from the optimal.
We also observe that we have (marginally) better performance in Fig.~\ref{fig:sqrt-real-data} for
$f(\mathcal{X})=\sqrt{|\mathcal{X}|}$ compared to $f(\mathcal{X})=\min\{3|\mathcal{X}|,10\}$ in Fig.~\ref{fig:min-real-data}.

\section{Conclusions and Future Research Directions}
\label{sec:conclusion}

\revision{In this work, we give an online competitive algorithm for adaptively batching samples that arrive over time.}
We consider the average \thirdrevision{waiting} 
time of the samples plus the average
per-sample processing cost as the objective to minimize. We first give a polynomial-time algorithm that computes
the optimal offline schedule. We then consider the online problem and devise constant-competitive algorithms for general submodular processing-cost functions. We also give a lower bound on the competitive ratio
that no online algorithm can beat.

However, there is still a gap between the competitive ratio guaranteed by our online algorithm and the proposed lower bound. 
\revision{Future research to close this gap might include (i) more precise analysis to show that \textsc{WTA} has a better competitive ratio,
(ii) better lower bounds that show no online algorithm can achieve an even higher competitive ratio, or
(iii) new algorithms that outperform \textsc{WTA} and achieve an improved competitive ratio.
Finally, as our numerical experiments demonstrate, the competitive ratio is just a worst-case measure, so utilizing additional statistics
of the arrival sequence might lead to improved performance in specific scenarios.}
We believe all these have much potential for future research.

\revision{Another interesting direction is to consider the case where there are deadlines for each sample. In this case, the objective remains the same, but the algorithm needs to ensure that the samples are processed before their respective deadlines. Alternatively, the algorithm could be penalized for each sample processed after its deadline. This requires algorithms that do not process all the samples available at a given time, but rather prioritize samples that are close to their deadlines, leaving other samples for later batches. This problem is more challenging than the one considered in this work, and it would be interesting to see if the techniques developed here can be extended to this setting. 
}
\secondrevision{
While minimizing the weighted sum of completion times is known to be NP-hard in this case \citep{LenstraKB1977}, 
it is unclear whether our objective of \thirdrevision{waiting} 
time plus processing cost can be minimized in polynomial time.
}



\bibliographystyle{elsarticle-harv} 
\bibliography{refs}

\newpage
\appendix

\medskip
\section{Linear program for the offline problem}
\label{sec:linear-program-offline}

Since we formulated the batch-scheduling problem as one of finding the minimum-weight
path in a graph, we can write it as an integer linear program.
However, as we show, we do not need to solve with integral constraints, and a relaxed version
of the program over a convex polyhedron gives the same cost as the integer program.
More concretely, consider the problem of finding the minimum-weight path in the graph
defined in Sec.~\ref{sec:offline} and define the binary variables $x_{i,j}\in\{0,1\}$ for $i<j$
with $i\in\{1,2,\ldots,n\}$
and $j\in\{2,3,\ldots,n+1\}$.
Each batching schedule corresponds to an assignment of $0$s and $1$s to all $x_{i,j}$,
and we interpret this as $x_{i,j}=1$ if and only if the arrivals $\{i,i+1,\ldots,j-1\}$ are
processed in a single batch.
Additionally, arrivals before $i$ and after $j-1$ are not part of this batch.%
\footnote{An assignment of $0$s and $1$s to $x_{i,j}$ that violates this is not a valid assignment.
As we show, solving the linear program always results in a valid assignment.}
Let the edge weights $e_{i,j}$ be defined as we did in Sec.~\ref{sec:offline} for $e_{ij}$,
i.e., $e_{i,j}:=f(j-i)+\sum_{k=i}^{j-1}(a_{j-1}-a_k)$.

Consider the following integer linear program.
\begin{align}
    \underset{x_{i,j}\in\{0,1\}}{\text{minimize}}&\quad \frac{1}{n}\sum_{i<j} x_{i,j}e_{i,j} \tag{ILP} \label{eq:ilp-offline} \\
    \text{subject to}&\quad \sum_{j=2}^{n+1} x_{1,j} = 1\qquad(\text{for}\ i=1),\nonumber \\
    &\quad \sum_{j=i+1}^{n+1} x_{i,j} = \sum_{j=1}^{i-1} x_{j,i} \quad \text{for all}\ i\in\{2,3,\ldots,n\}.\nonumber
\end{align}
\begin{claim}
    The optimal solution of \eqref{eq:ilp-offline}, $\{x_{i,j}^*\}$, has a cost equal to $J^\textsc{opt}$, the cost
    of the optimal batching schedule.
    Further, this cost can be achieved by scheduling in such a way that we form a batch
    of arrivals $\{i,i+1,\ldots,j-1\}$ if and only if $x_{i,j}^*=1$.
    Solving \eqref{eq:ilp-offline} results in a solution where this batching gives a valid
    schedule (i.e., no arrival is assigned to multiple batches, and every arrival is assigned to at least one batch).
\label{claim:ilp-offline}
\end{claim}
\begin{proof}
We prove this claim in two parts:
(i) we show that there is a feasible solution which gives a cost of $J^\textsc{opt}$ ---
this guarantees that the optimal cost is not more than $J^\textsc{opt}$; and
(ii) we show that any (optimal) solution of \eqref{eq:ilp-offline} corresponds to a valid batching
schedule.
This implies that the optimal solution cannot have a cost less than $J^\textsc{opt}$,
since the the best possible schedule for the batching problem has a cost of $J^\textsc{opt}$.

For part (i), let a best batching schedule be given by the $m$ batches
$\mathcal{B}^\textsc{opt}_1,\mathcal{B}^\textsc{opt}_2,\ldots,\mathcal{B}^\textsc{opt}_m$.
For each $\mathcal{B}^\textsc{opt}_k$, define $p_k:=\min\{i:i\in\mathcal{B}^\textsc{opt}_k\}$
and $q_k:=\max\{i:i\in\mathcal{B}^\textsc{opt}_k\}+1$, and set $x_{p_k,q_k}=1$.
Set the remaining $x_{i,j}=0$.
From the definition of $e_{i,j}$, we can see that the cost of this schedule is equal to $\frac{1}{n}\sum_{i<j}x_{i,j}e_{i,j}$.
To show that this assignment of $0$s and $1$s to $\{x_{i,j}\}$ is feasible, observe the following.
Since $\mathcal{B}^\textsc{opt}_1$ includes the first arrival, the value of $x_{1,\max\{\mathcal{B}^\textsc{opt}_1\}+1}$
would be set to $1$.
Further, no other batch contains the first arrival, so the sum $\sum_{j=2}^{n+1}x_{1,j}$ is equal to $1$.
This is sufficient to satisfy the first ($i=1$) constraint.
For each $i\in\{2,3,\ldots,n\}$, the only way the right side of the second constraint is nonzero is
if the $(i-1)$th arrival was the last arrival of some batch $\mathcal{B}^\textsc{opt}_k$.
In this case, the right side is equal to $1$, but the left side is also equal to $1$ since 
the $i$th arrival is now the first sample of $\mathcal{B}^\textsc{opt}_{k+1}$.
So all the constraints are satisfied and this solution is feasible.
Thus the optimal cost of \eqref{eq:ilp-offline} is not more than $J^\textsc{opt}$.

Now we have to prove part (ii).
Assume that the solution of \eqref{eq:ilp-offline} is $\{x_{i,j}^*\}$,
and we batch the arrivals into batches $\{i,i+1,\ldots,j-1\}$ if and only if $x_{i,j}^*=1$.
Fixing an $i=\hat{i}\ (\in\{1,2,\ldots,n\})$ and adding up the constraints for $i=1,2,\ldots,\hat{i}$ gives%
\footnote{For $\hat{i}=1$, the summation on the right side is over an empty set, and is taken to be $0$.}
\begin{align*}
    \sum_{i=1}^{\hat{i}}\sum_{j=i+1}^{n+1}x_{i,j} = 1 + \sum_{i=1}^{\hat{i}}\sum_{j=1}^{i-1}x_{j,i}.
\end{align*}
Every $x_{i,j}$ term on the right side in the above equation also occurs on the left side, and cancelling these gives
\begin{align}
    \sum_{i=1}^{\hat{i}}\sum_{j=\hat{i}+1}^{n+1}x_{i,j} = 1\quad \forall\ \hat{i}\in\{1,2,\ldots,n\}.
\label{eq:cross-section-is-1}
\end{align}
Since $\{x_{i,j}^*\}$ (optimally) solves \eqref{eq:ilp-offline}, it must satisfy \eqref{eq:cross-section-is-1}.
Since all $x_{i,j}$ are binary variables, \eqref{eq:cross-section-is-1} implies for any $\hat{i}$ that
there is exactly one $(i,j)$ with $i\le\hat{i}$ and $j>\hat{i}$ such that $x_{i,j}=1$.
So each arrival $\hat{i}$ is assigned to exactly one and only one batch.
This proves that the optimal solution $\{x_{i,j}^*\}$ always results in a valid batching schedule
when we schedule it according to the scheme given in the claim, which concludes the proof.
\end{proof}

Solving \eqref{eq:ilp-offline} directly using integer linear-programming solvers
for large problem instances is intractable due to
the integer constraints $x_{i,j}\in\{0,1\}$.
However, we can relax these constraints to the convex $x_{i,j}\ge0$, and it turns out that there is an optimal solution
to this (convex) linear program where $x_{i,j}$ is binary.%
\footnote{If a solver returns a solution that is not binary, it indicates some
degeneracy in the problem with multiple optimal solutions.
This can be fixed by randomly perturbing $\{e_{i,j}\}$ by an infinitesimal amount.}
Since the domain of $\{x_{i,j}\ge0\}$ is a superset of $x_{i,j}\in\{0,1\}$,
this solution of the relaxed linear program is thus an optimal solution of \eqref{eq:ilp-offline} as well.
We state this relaxed program formally as \eqref{eq:ilp-offline-relaxed} and the equivalence as Claim~\ref{claim:ilp-offline-relaxed}.
\begin{align}
    \underset{x_{i,j}\ge0}{\text{minimize}}&\quad \frac{1}{n}\sum_{i<j} x_{i,j}e_{i,j} \tag{LP} \label{eq:ilp-offline-relaxed} \\
    \text{subject to}&\quad \sum_{j=2}^{n+1} x_{1,j} = 1\qquad(\text{for}\ i=1),\nonumber \\
    &\quad \sum_{j=i+1}^{n+1} x_{i,j} = \sum_{j=1}^{i-1} x_{j,i} \quad \text{for all}\ i\in\{2,3,\ldots,n\}.\nonumber
\end{align}
\begin{claim}
    There is an optimal solution of \eqref{eq:ilp-offline-relaxed}, $\{\tilde{x}_{i,j}\}$,
    that satisfies $\tilde{x}_{i,j}\in\{0,1\}$.
\label{claim:ilp-offline-relaxed}
\end{claim}
\begin{proof}
For each constraint in \eqref{eq:ilp-offline-relaxed}, define the dual variable $\lambda_i$ for $i\in\{1,2,\ldots,n\}$.
The dual program of \eqref{eq:ilp-offline-relaxed}, \eqref{eq:ilp-offline-relaxed-dual} is given by
\begin{align}
    \underset{\lambda_i\in\mathbb{R}}{\text{maximize}}&\qquad \lambda_1 \tag{DP}\label{eq:ilp-offline-relaxed-dual} \\
    \text{subject to}&\quad \lambda_i \le \frac{e_{i,n+1}}{n}\quad \text{for}\ i\in\{1,2,\ldots,n\},\nonumber\\
    &\quad \lambda_i \le \frac{e_{i,j}}{n} + \lambda_j\quad \text{for}\ i<j\ \ (i,j\in\{1,2,\ldots,n\}).\nonumber
\end{align}
We prove Claim~\ref{claim:ilp-offline-relaxed} by showing that the objective value of a dual-feasible
solution is equal to that achieved by a binary solution for the primal program \eqref{eq:ilp-offline-relaxed}.
To see this, define $\lambda_{n+1}:=0$ so we can write both the constraints compactly as
\begin{align*}
    \lambda_i \le \frac{e_{i,j}}{n} + \lambda_j\quad \text{for} \ i<j\ \ (i,j\in\{1,2,\ldots,n+1\}).
\end{align*}
Since we are maximizing $\lambda_1$ in \eqref{eq:ilp-offline-relaxed-dual}, it will be equal to the right side of one of the inequalities
containing $\lambda_1$:
\begin{align*}
    \lambda_1 = \min_{j\in\{2,3,\ldots,n+1\}}\left\{\frac{e_{1,j}}{n} + \lambda_j\right\}.
\end{align*}
This implies that for the maximum value of $\lambda_1$, we need to find the maximum values
of $\lambda_j$ for $j\in\{2,3,\ldots,n\}$ (the value of $\lambda_{n+1}$ is fixed to $0$).
This suggests a natural recursion for computing the optimal value of $\lambda_1$.
Starting with the base case $\lambda_{n+1}=0$, we have for $i\in\{n,n-1,\ldots,1\}$:
\begin{align}
    \lambda_i = \min_{j\in\{i+1,i+2,\ldots,n+1\}}\left\{\frac{e_{i,j}}{n}+\lambda_j\right\}.
    \label{eq:lambda-recursion}
\end{align}
Let $r_i$ denote the index where the minimum in \eqref{eq:lambda-recursion} occurs:
\begin{align*}
    r_i = \arg\min_{j\in\{i+1,i+2,\ldots,n+1\}}\left\{\frac{e_{i,j}}{n}+\lambda_j\right\}.
\end{align*}
Use the following procedure to assign values to $\{\tilde{x}_{i,j}\}$:
\begin{algorithmic}[1]
    \State Initialize $i\gets1$ and $\tilde{x}_{p,q}\gets0$ for all $p<q$
    \While {$i \le n$}
        \State $\tilde{x}_{i,r_i}\gets1$
        \State $i\gets r_i$
    \EndWhile
\end{algorithmic}
From the definition of $r_i$, we have
\begin{align*}
    \lambda_i = \frac{e_{i,r_i}}{n} + \lambda_{r_i},
\end{align*}
which implies
\begin{align}
    \lambda_1 = \frac{1}{n}\sum_{i=1}^{n}\sum_{j=i+1}^{n+1}\tilde{x}_{i,j}e_{i,j}.
    \label{eq:primla-equals-dual}
\end{align}
Thus there is a dual-feasible solution which has the same cost as a binary solution for
the primal program \eqref{eq:ilp-offline-relaxed}.
We now just have to show that this binary solution satisfies all the constraints of \eqref{eq:ilp-offline-relaxed}
to conclude the proof.
The first constraint of \eqref{eq:ilp-offline-relaxed} is satisfied since $\tilde{x}_{1,r_1}=1$
and $\tilde{x}_{1,j}=0$ for $j\neq r_1$.
Further, observe that $\tilde{x}_{i,j}$ is nonzero if and only if $i$ and $j$ both occur in the
$\{r_1,r_{r_1},r_{r_{r_1}},\ldots\}$ sequence, and $j=r_i$.
This ensures that for each element in this sequence, there is exactly one term in the second constraint
of \eqref{eq:ilp-offline-relaxed} that is $1$.
All other terms are $0$ guaranteeing that the constraint is valid for all $i$, and thus $\{\tilde{x}_{i,j}\}$
is a primal-feasible solution.
So there is a feasible binary solution for the primal which has an objective equal
to a feasible solution for the dual (by \eqref{eq:primla-equals-dual}), which means
that these are the optimal solutions for both the primal and the dual \citep{BoydV2004}.
\end{proof}

Together with Claim~\ref{claim:ilp-offline}, Claim~\ref{claim:ilp-offline-relaxed} implies that
we can solve \eqref{eq:ilp-offline-relaxed} to find the optimal batching schedule.
This can be done using any off-the-shelf linear-program solver since \eqref{eq:ilp-offline-relaxed}
is just a standard linear program.
We also observe in the proof of Claim~\ref{claim:ilp-offline-relaxed} that solving the dual program bears a lot of similarity to how the minimum-weight path is computed once we have a topologically sorted order of the vertices in the graph.
The recursion in \eqref{eq:lambda-recursion} is essentially the same dynamic program that 
computes the minimum-weight path via topological sorting \citep{CormenLRS2022}.





\end{document}